\documentclass[letterpaper, 10 pt, conference]{ieeeconf}
\IEEEoverridecommandlockouts
\IEEEoverridecommandlockouts

\usepackage{cite}
\usepackage{graphicx}
\usepackage{picinpar}
\usepackage{amsmath}
\usepackage{stfloats}
\usepackage{url}
\usepackage{flushend}
\usepackage[latin1]{inputenc}
\usepackage{colortbl}
\usepackage{soul}
\usepackage{multirow}
\usepackage{pifont}
\usepackage{color}
\usepackage{alltt}
\usepackage{enumerate}
\usepackage{siunitx}
\usepackage{pbox}
\usepackage{amsmath,amssymb,amsfonts}
\usepackage{amsthm}
\usepackage{algorithmic}
\usepackage{algorithm}
\usepackage{graphicx}
\usepackage{graphbox}
\usepackage{textcomp}
\usepackage{xcolor}
\usepackage[export]{adjustbox}
\usepackage{color}
\usepackage{tikz}
\def\BibTeX{{\rm B\kern-.05em{\sc i\kern-.025em b}\kern-.08em
    T\kern-.1667em\lower.7ex\hbox{E}\kern-.125emX}}

\definecolor{LightCyan}{rgb}{1,0.8,0.8}

\newtheorem{theorem}{Theorem}
\newtheorem{corollary}{Corollary}

\makeatletter
\let\NAT@parse\undefined
\makeatother
\usepackage{hyperref}

\newcommand\copyrighttext{%
  \footnotesize \textcopyright This paper has been accepted for publication in the IEEE International Conference on Robotics and Automation. Please cite the paper as: E. Sebasti\'{a}n and E. Montijano,``Multi-robot Implicit Control of Herds'', IEEE International Conference on Robotics and Automation (ICRA), 1601-1607, 2022.}
\newcommand\copyrightnotice{%
\begin{tikzpicture}[remember picture,overlay]
\node[anchor=south,yshift=10pt] at (current page.south) {\fbox{\parbox{\dimexpr\textwidth-\fboxsep-\fboxrule\relax}{\copyrighttext}}};
\end{tikzpicture}%
}

\title{\LARGE\bf Multi-robot Implicit Control of Herds}

\author{\centering Eduardo Sebasti\'{a}n and Eduardo Montijano
\thanks{E. Sebasti\'{a}n and E. Montijano are with the RoPeRt group, at DIIS - I3A, Universidad de Zaragoza, Spain.
\texttt{\small \{esebastian, emonti\}@unizar.es}}
\thanks{This work has been supported by the ONR Global
grant N62909-19-1-2027, the Spanish projects PGC2018-098817-A-I00 and PGC2018-098719-B-I00 (MCIU/AEI/FEDER, UE), DGA T04-FSE, and Spanish grant FPU19-05700.}
}

\begin{document}
\maketitle

\copyrightnotice


\begin{abstract}
This paper presents a novel control strategy to herd a group of non-cooperative evaders by means of a team of robotic herders. In herding problems, the motion of the evaders is typically determined by strong nonlinear reactive dynamics, escaping from the herders. Many applications demand the herding of numerous and/or heterogeneous entities, making the development of flexible control solutions challenging. In this context, our main contribution is a control approach that finds suitable herding actions even when the nonlinearities in the evaders' dynamics yield to implicit equations. We resort to numerical analysis theory to characterise the existence conditions of such actions and propose two design methods to compute them, one transforming the continuous time implicit system into an expanded explicit system, and the other applying a numerical method to find the action in discrete time. Simulations and real experiments validate the proposal in different scenarios.
\end{abstract}


\section{Introduction}\label{sec:intro}
Recent advances in Multi-Robot Systems (MRS) have favoured the development of successful control strategies in real-life problems such as entrapment~\cite{Antonelli_RAM_2008_Entrapment}, hunting~\cite{Zhu_IJARS_2015_Hunting} or escorting~\cite{Gao_ACCESS_2018_Escorting}. Despite the different nature of scenarios, these problems can be gathered as herding~\cite{Pierson_2018_TR_Herding}, where the objective is to drive a group of targets to specific locations using a team of robots. A common denominator is the non-cooperative nature of the targets with respect to the control objective, typically entangled in complex nonlinear behaviours. Indeed, the difficulties hidden in the herding problem have motivated broader interdisciplinary research gathering physiologists, mathematicians and neurologists with engineers~\cite{Strombom2014Shepherding}~\cite{Long2020Comprehensive}.  
To cope with this, we present a control solution that, relying on numerical analysis, solves a set of implicit equations and drives the herd towards the goal. Our method works with different motion models and team configurations.

\begin{figure}[!ht]
\centering
\begin{tabular}{cc}
     {\footnotesize (a)}
     &
     {\footnotesize (b)}
     \\\hspace{0.1cm}
     \includegraphics[width=0.41\columnwidth,height=0.24\columnwidth]{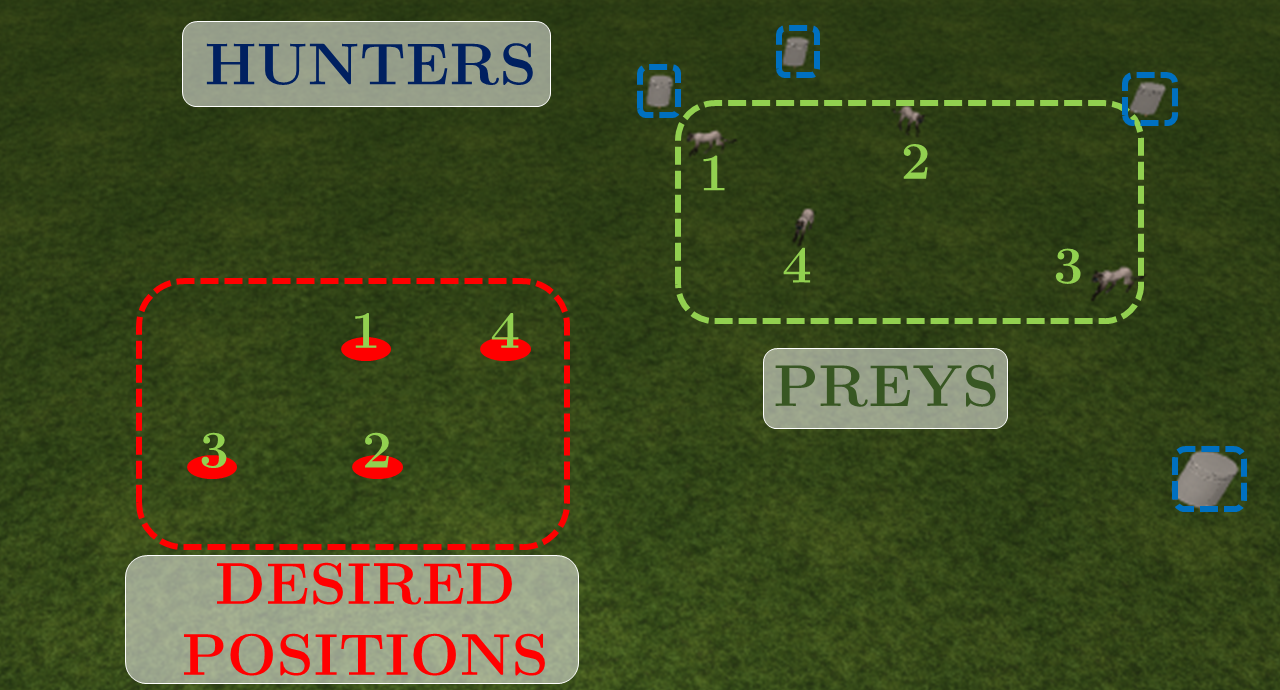}\hspace{0.1cm}
     &  \hspace{0.1cm}
     \includegraphics[width=0.41\columnwidth,height=0.24\columnwidth]{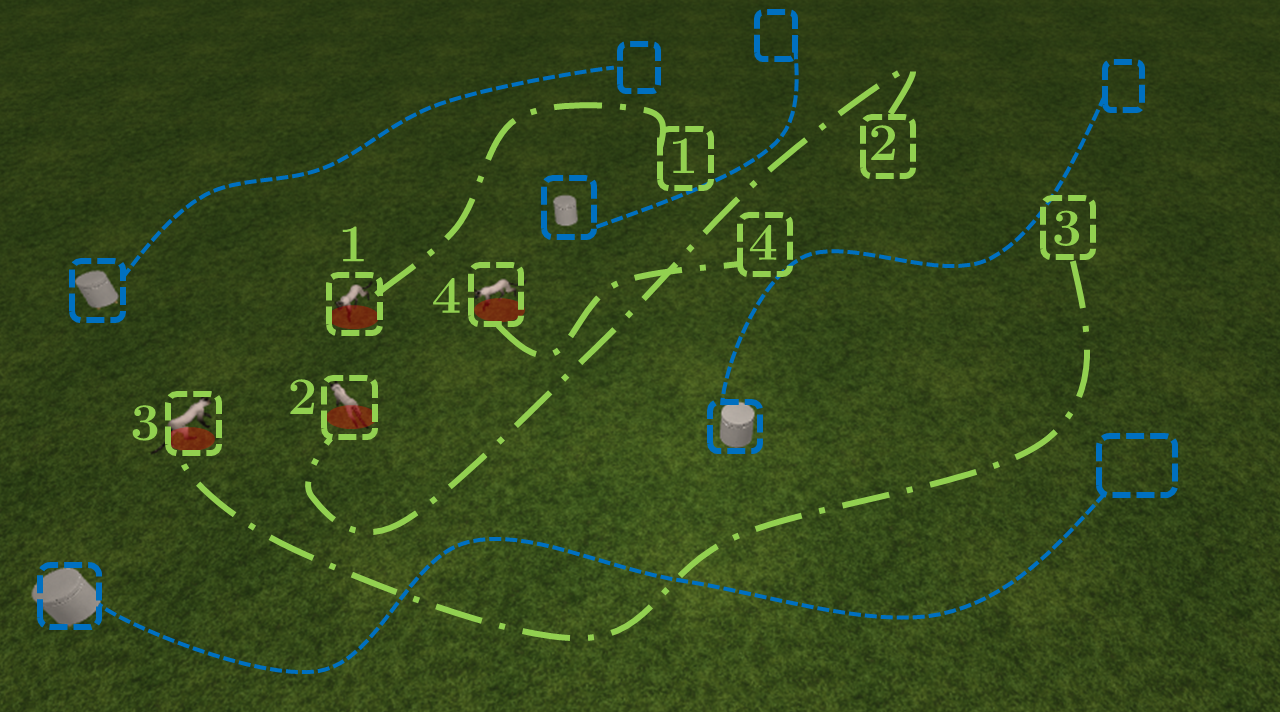}\hspace{0.1cm}
     \\
     {\footnotesize (c)}
     &
     {\footnotesize (d)}
     \\\hspace{0.1cm}
     \includegraphics[width=0.41\columnwidth,height=0.24\columnwidth]{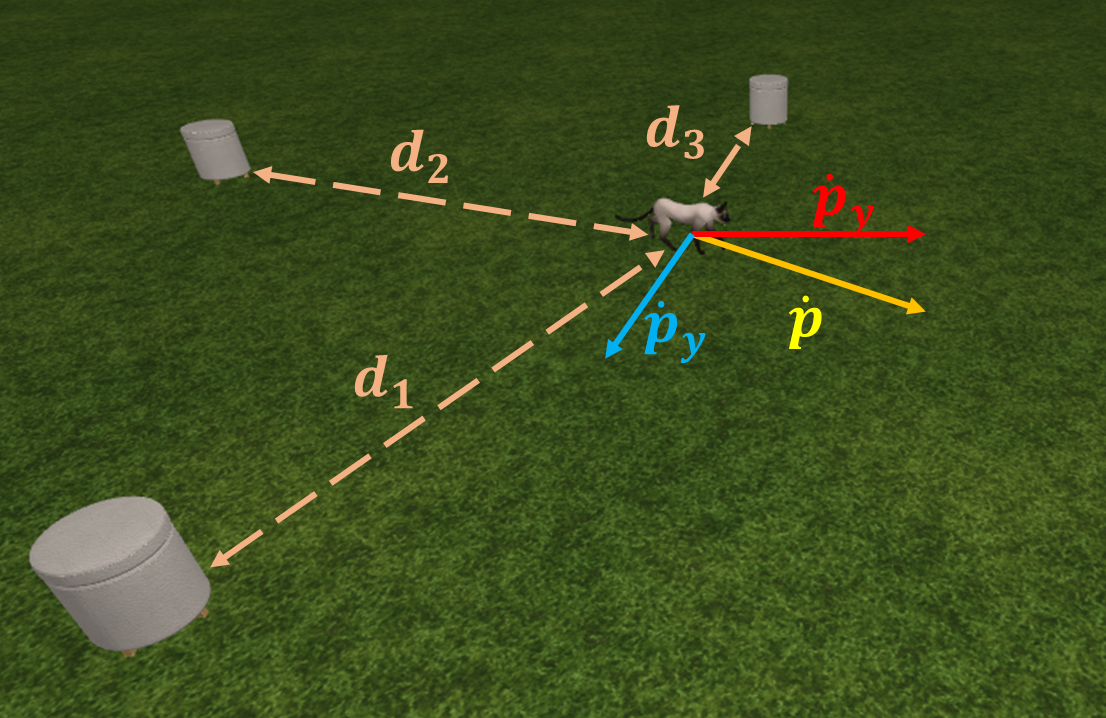}\hspace{0.1cm}
     & \hspace{0.1cm}
     \includegraphics[width=0.41\columnwidth,height=0.24\columnwidth]{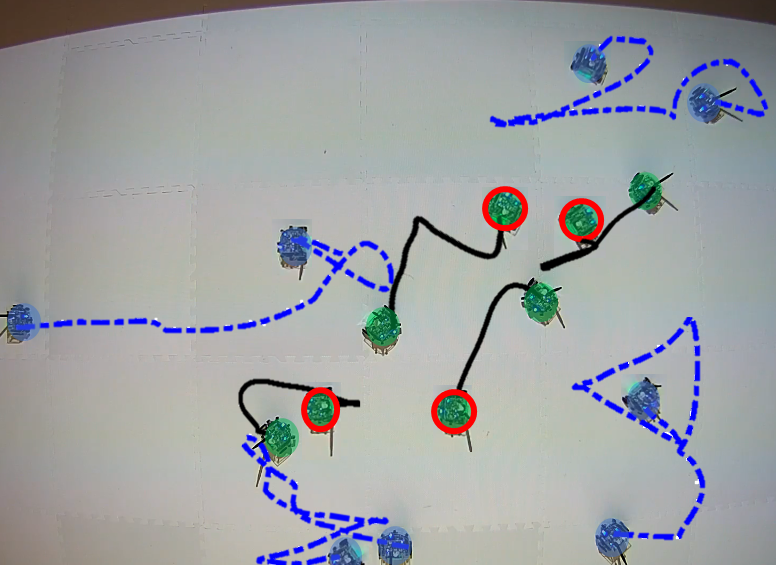}\hspace{0.1cm}
\end{tabular}
    
\caption{In the herding (a), the robotic herders (blue) drive the evaders (green) to specific locations (red) simultaneously. Once the herders are near the evaders and they are assigned to a desired location (b), the herding starts. The control strategy leverages the repulsive interaction forces (c) to locate the evaders in their corresponding desired places. The view in (d) visualise an experiment where four robots herd four evaders, showing the initial and final positions, desired positions with red contour, and trajectories of the players. A video of the experiments is included as supplementary material.}
\label{fig:first_impression}
\end{figure}

In the literature, most of the works deal with the herding of a single entity. An early example is~\cite{Antonelli_RAM_2008_Entrapment}, which employs a Null space-based behavioral control to escort/entrap the target. Meanwhile,~\cite{Monti_CDC_2013_Entrapment} deals with the uncertainty in the position of the target by adopting an elliptical orbit around it. In a different fashion, the authors of~\cite{Zhu_IJARS_2015_Hunting} apply a Bio-inspired Neural Network to hunt in an underwater environment. Conversely, the approach in~\cite{Belkhouche_ICSMC_2005_Hunting} relies on geometric rules in a land-based hunting. Another example is~\cite{Gao_ACCESS_2018_Escorting}, where a vector-field based controller escorts an objective. From this point of view, only a few papers solve the entrapment of more than one target. In~\cite{Jahn_ICRA_2017_Surveillance}, a group of robots navigates around a certain area to avoid targets to cross its boundaries. The work in~\cite{Pierson_2018_TR_Herding} drives groups of entities by an active encirclement but does not consider specific final positions for each evader. Following a similar approach, the authors in~\cite{Panagou_CDC_2019_Herding} go a step further and propose an active encirclement which avoids obstacles. In~\cite{Licitra_CSL_2018_Herding}, a single robot is in charge of herding multiple targets controlling them one by one. In contrast, our method performs the herding of any number of targets to precise individual locations simultaneously.

Another feature of herding problems is the behaviour of the targets, where two assumptions are often considered: linear and homogeneous dynamics. The first assumption is considered in, e.g.,~\cite{Jiang_TAC_2019_Containment}, where the containment of linear heterogeneous agents is performed by a time-varying formation. The authors in~\cite{Scott_ACC_2013_PHE} apply control theory in a problem with three linear entities: a pursuer, a herder and an evader. The work in~\cite{Anisi_TAC_2010_Surveillance} drives a group of linear UGVs to ensure that a region is completely surveyed whilst in~\cite{Ramana_CDC_2015_PE} both pursuers and evader have the same linear dynamics, solving the problem with geometric tools. The second assumption is used in~\cite{Pierson_2018_TR_Herding}, where nonlinear homogeneous repulsive dynamics are controlled using a team of robots. A different instance is~\cite{Alexopoulos_ICUAS_2017_PE}, where a complete control structure is presented, from the hex-rotor motion to the pursuit layer. Our proposal offers a general framework to design a control strategy for the herding of general nonlinear and heterogeneous targets.

Summarising, the main contribution of this paper is a novel control strategy for the herding in MRS. In contrast to the state of the art, our approach is valid for any number of targets, as well as heterogeneous dynamics under two mild assumptions. The first assumption is that the maximum velocity of herders and evaders is equal, making the herding even more challenging. Derived from this first assumption, an initial approaching phase is supposed to be done to surround the evaders, to avoid their evasion before beginning the herding. The second is that a target assignment is already done to assign desired positions to evaders. We leverage numerical analysis theory to find the herding action when, due to complex nonlinearities, the control is defined by a set of implicit equations. Two design procedures are derived, an explicit version that expands the system with appropriate action dynamics, and an implicit version which leverages existence and smoothness properties to solve the action with a numerical method. Simulated and real experiments demonstrate the success of the solution, using two state-of-the-art dynamic models for the evaders to validate the flexibility of the proposal.

The rest of the paper is arranged as follows. Section~\ref{sec:prosta} establishes the bases of the problem. Section~\ref{sec:IC} develops the methodologies to calculate the control action. Section~\ref{sec:simulations} discusses the simulation results.
The success of the proposal is corroborated with real experiments in Section~\ref{sec:experiments}. Finally, Section~\ref{sec:conclusion} contains the conclusions.


\section{Problem Statement}\label{sec:prosta}
We consider the problem of herding a group of $m$ evaders using a team of $n$ robotic herders. We denote evaders by $\textit{j} \in \{ 1, 2,  ... , m \},$ herders by $\textit{i} \in \{1, 2, ... , n\}$, and we define the state $\mathbf{x} \in X \subseteq \mathbb{R}^{2m}$ and the action $\mathbf{u} \in U \subseteq \mathbb{R}^{2n}$ as
\begin{align*}
\mathbf{x} = \begin{bmatrix}
    \mathbf{p}_{1}^{T} & \mathbf{p}_{2}^{T} & \dots  & \mathbf{p}_{m}^{T}
    \end{bmatrix}^{T},
& \:
\mathbf{u} = \begin{bmatrix}
    \mathbf{h}_{1}^{T} & \mathbf{h}_{2}^{T} & \dots  & \mathbf{h}_{n}^{T}
    \end{bmatrix}^{T}, 
\end{align*}
where $\mathbf{p}_{j} \in \mathbb{R}^{2}$ is the position of evader $j$ and $\mathbf{h}_{i} \in \mathbb{R}^{2}$ is the position of herder $i$, operating in a 2D space\footnote{The choice of a 2D space is to fit the real experiments, but the solution can be generalised to higher dimensions without changes in the formulation.}. Then, we define the movement of each evader by
\begin{equation}
\dot{\mathbf{p}}_j = f_j(\mathbf{x},\mathbf{u}) 
\label{eq:base}
\end{equation}
allowing for any nonlinear behaviour encoded in $f_j(\mathbf{x},\mathbf{u})$. 
The only assumption regarding $f_j$ is that it is continuous and derivable, i.e., of class $C^1$, for all $j$. To ease the intuition and to assess the proposal against heterogeneous evaders, we exemplify Eq.~\eqref{eq:base} using 
two dynamic models from the literature.
Their choice is motivated by the strong nonlinear behaviour in the position of evaders and herders.

The \textit{Inverse Model} (extracted from \cite{Pierson_2018_TR_Herding}) is given by
\begin{equation}
\dot{\mathbf{p}}_j = f_j^{inv}(\mathbf{x},\mathbf{u}) = \gamma_j\sum_{i=1}^{n} \frac{\mathbf{d}_{ij}}{||\mathbf{d}_{ij}||^{3}},
\label{eq:PiersonBase}
\end{equation}
where $\mathbf{d}_{ij} = \mathbf{p}_{j} - \mathbf{h}_{i}$ is the relative position between evader $j$ and herder $i$ and $\gamma_j$ is a positive constant which can express different phenomena such illness, loss of energy, etc.

We note that despite the model has a singularity in $\mathbf{d}_{ij}=\mathbf{0}$, the speed remains bounded. The repulsion grows with $\frac{1}{||\mathbf{d}_{ij}||^{3}}$ so the closer the herders, the larger the repulsion. Since the maximum speed of the players is equal, the only way of achieving $\mathbf{d}_{ij}=\mathbf{0}$ is that an evenly spaced number of herders approaches the evader with symmetric trajectories such that the sum of the repulsion forces is $\mathbf{0}$.

The \textit{Exponential Model} (extracted from \cite{Licitra_CSL_2018_Herding}) is given by  
\begin{equation}
\label{eq:LicitraBase}
\dot{\mathbf{p}}_j = f_j^{exp}(\mathbf{x},\mathbf{u}) = \left\{
\begin{array}{ll}
\displaystyle\beta_j\alpha_j\sum_{i=1}^{n}\mathbf{d}_{ij}e^{-\chi_{ij}},&\hbox{ if } ||\mathbf{d}_{ij}|| > r 
\\
\displaystyle\alpha_j\sum_{i=1}^{n}\mathbf{d}_{ij}e^{-\chi_{ij}},&\hbox{ otherwise,}
\end{array}
\right.
\end{equation}
where $\chi_{ij} = \frac{1}{\sigma^2}\mathbf{d}_{ij}^{T}\mathbf{d}_{ij}$ and $\sigma_j>1$. In this model there is a switching condition if $||\mathbf{d}_{ij}|| \leq r$, where the evader $j$ becomes ``scared'' and the intensity of the repulsive interaction increases, due to $0<\beta_j<1$. Besides, $\alpha_j>0$ works as $\gamma_j$ does in the Inverse Model.

The joint system dynamics of Eq.~\eqref{eq:base} can be defined as
\begin{equation}
\dot{\mathbf{x}} = f(\mathbf{x},\mathbf{u})
\label{eq:initial_system},
\end{equation}
where $f(\mathbf{x},\mathbf{u})$ simply comes from stacking all $f_j(\mathbf{x},\mathbf{u})$.
This formulation also allows to consider heterogeneous herds, with different number of evaders and motion models. 

Our goal is to herd the evaders to a desired set of positions $\mathbf{x}^{*}\in X$ simultaneously. 
To do this, we define the position error of the evaders as $\widetilde{\mathbf{x}} = \mathbf{x} - \mathbf{x}^{*}$ and we set the control objective as to find a control strategy to drive $\widetilde{\mathbf{x}}$ to zero. 

It is noteworthy that the reactive behaviour of the evaders is with respect to the position of the herders. Therefore, a control strategy which determines actions in terms of position of the herders is adequate to generalise the solution to different robotic platforms. 
Besides, in this problem we consider that full state feedback is available.

We do two considerations before ending this Section. To keep the generality of the solution, in this work we assume that the maximum velocity of herders and evaders is equal. This requires an initial approaching phase, surrounding the evaders to avoid their escape before the precision herding begins. 
It is not mandatory to achieve a compact and closed encirclement to succeed but just the distribution of the herders near the evaders. Examples of approximation and encirclement that could be used are~\cite{Antonelli_RAM_2008_Entrapment,Monti_CDC_2013_Entrapment,Panagou_CDC_2019_Herding} or \cite{Zhang2018Swarm}. Similarly, a simple task assignment associates each evader with its corresponding desired location, e.g., with the Hungarian algorithm, reducing this way the chances of collision caused by crossing paths. 


\section{Implicit Control}\label{sec:IC}
The herding seeks an expression for the position of the herders, i.e., for the action $\mathbf{u}$, such that the evaders go to their assigned desired positions, $\mathbf{x}^*.$
Besides, the process may need to accomplish some other requirements, such as a desired transient response. These objectives can be translated into designing $\mathbf{u}$ such that the evaders follow desired dynamics $f^*$ which fulfil the requirements,
\begin{equation}\label{eq:def_f_star}
    \dot{\mathbf{x}} = f^*(\mathbf{x}).
\end{equation}

Instead of looking for a closed-form expression for the action, $\mathbf{u}=g(\mathbf{x})$, we propose to compute the action by straight substitution of the evader dynamics for the desired ones, namely, find the value of $\mathbf{u}$ that transforms~\eqref{eq:initial_system} into~\eqref{eq:def_f_star}.
However, trying to do this in complex nonlinear models such the ones described in Section~\ref{sec:prosta} results in systems of implicit equations,
where finding a closed-form solution is not possible nor analytically tractable.
To overcome this, in this Section we study the conditions that allow to find $\mathbf{u}$ such that the evaders evolve according to $f^*$. Then, we propose two design procedures to solve the control action. The description is kept in general control terms since we believe that this procedure can be of interest in other control problems.

Firstly, it is necessary to address the question of whether a smooth action that makes the actual dynamics equal to the desired ones exists or not.
To do this, we define
\begin{equation}\label{eq:def_h}
    h(\mathbf{x},\mathbf{u}) = f(\mathbf{x}, \mathbf{u}) - f^*(\mathbf{x}),
\end{equation}
changing the framework to that of computing the roots of $h(\mathbf{x},\mathbf{u})$ with respect to $\mathbf{u}$.
With this change we can use the Implicit Function Theorem
to formally characterise sufficient conditions for existence and smoothness.
\begin{theorem}[Adapted from Theorem 9.28 of~\cite{Rudin1976Mates}]
\label{theorem:exist_and_smooth}
Let 
\begin{equation}\label{eq:formal_h}
    h: I = X \times U \subset\mathbb{R}^{n}\times \mathbb{R}^{m}\longmapsto\mathbb{R}^{m}
\end{equation}
a $C^1$-mapping, such that $h(\mathbf{x}^*,\mathbf{u}_0^{*})=\mathbf{0}$ for some point $\mathbf{u}_0^{*} \in U$. Additionally, consider the Jacobian
\begin{equation}\label{eq:J_in_depth}
    \mathbf{J} = \left(\mathbf{J}_{\mathbf{x}} | \mathbf{J}_{\mathbf{u}})\right. = 
    \left.\begin{pmatrix}
    \frac{\partial h_1}{\partial \mathbf{x}_1} & \hdots & \frac{\partial h_1}{\partial \mathbf{x}_m}
    & 
    \frac{\partial h_1}{\partial \mathbf{u}_1} & \hdots & \frac{\partial h_1}{\partial \mathbf{u}_n}
    \\
    \vdots & \ddots & \vdots & \vdots & \ddots & \vdots
    \\
    \frac{\partial h_m}{\partial \mathbf{x}_1} & \hdots & \frac{\partial h_m}{\partial \mathbf{x}_m}
    &
    \frac{\partial h_m}{\partial \mathbf{u}_1} & \hdots & \frac{\partial h_m}{\partial \mathbf{u}_n}
    \end{pmatrix}\right..
\end{equation}
such that $\mathbf{J}_{\mathbf{u}}$ is non-singular in the point $(\mathbf{x}^*,\mathbf{u}_0^{*})$. Then, there exist open subsets $I^*\subset\mathbb{R}^{n}\times \mathbb{R}^{m}$ and $X^*\subset\mathbb{R}^{n}$, with  $(\mathbf{x},\mathbf{u}^{*}) \in I^*$ and $\mathbf{x}\in X^*$, having the following property: to every possible $\mathbf{x} \in X^*$ corresponds a unique $\mathbf{u}^*$ such that $(\mathbf{x},\mathbf{u}^{*}) \in I^*$ and $h(\mathbf{x},\mathbf{u}^*)=\mathbf{0}$.
\end{theorem}

The Theorem imposes three conditions to be fulfilled.
Firstly, there must exist an action, $\mathbf{u}_0^{*},$ which solves the control in $\mathbf{x}^{*}$. In the herding context, there must exist a stable configuration of the herders when the evaders are in their desired positions. A sufficient condition to ensure this, for all $\mathbf{x} \in X$, is to let $n \geq m,$ since both evaders and herders are first order entities in the space.
Nevertheless, there might be configurations for which a fewer number of herders is enough.

Secondly, $h$ must be of class $C^1$ in $(\mathbf{x}^*,\mathbf{u}_0^{*})$. If $f^*$ is chosen of class $C^1$ in $\mathbf{x}^*$, then, this condition is accomplished because $f_j$ in~\eqref{eq:base} is of class $C^1$ for all $j$ by assumption, so $f$ in~\eqref{eq:initial_system} is also of class $C^1$.

The last condition requires the Jacobian of $h$ with respect to $\mathbf{u}$, $\mathbf{J}_{\mathbf{u}}$, to be non-singular in the desired location. Since for $m\neq n$ the matrix is not square, we generally consider the matrix $\mathbf{J}_{\mathbf{u}}^T \mathbf{J}_{\mathbf{u}}$ as the one to be non-singular. Given the aforementioned features of $h$ and $\mathbf{J}_{\mathbf{u}}$, the last condition is accomplished in $\mathbf{x}^*$. Moreover, restricting $I$ to the subspace without collisions we ensure that the two last conditions of Theorem~\ref{theorem:exist_and_smooth} hold for all $\mathbf{x}$ and, therefore, the Theorem holds for all $\mathbf{x}$ in this subspace.

Considering that each herder provokes a repulsive reaction in every evader, collision of herders and evaders will not happen in practice. Similarly, given that we are controlling the herders, it is easy to prevent collisions among them.

\begin{corollary}\label{corollary:implicit_theorem}
If the conditions in Theorem~\ref{theorem:exist_and_smooth} hold $\forall \mathbf{x}\in X$, the existence and smoothness of control action is globally guaranteed in $I$.
\end{corollary}
In the next Subsections we present two methods to compute the action $\mathbf{u}$ that solves $h$, i.e., which imposes the desired behaviour of the evaders.

\subsection{Explicit Design}\label{subsec:explicit_design}

The Explicit Design method consists in expanding the initial system in~\eqref{eq:initial_system} with action dynamics that converge to the roots of $h$. This transforms the problem to that of computing the action $\mathbf{u}$ as part of an expanded explicit system, described in continuous time and with analytical solution.

To do this, we propose a design on the action dynamics
\begin{equation}
\label{eq:u_dynamics}
\kern -4pt
    \dot{\mathbf{u}} = \mathbf{J}_{\mathbf{u}}^{T} \mathbf{J}_{\mathbf{u}}^{+}
    \kern -3pt\left(h^*(\mathbf{x},\mathbf{u})-\mathbf{J}_{\mathbf{x}} f(\mathbf{x}, \mathbf{u}) \right),
\end{equation}
where $h^*$ is a free design parameter that encodes the desired closed-loop dynamic behaviour of $h$,
$\mathbf{J}_{\mathbf{u}}^{+}  = (\mathbf{J}_{\mathbf{u}}^T \mathbf{J}_{\mathbf{u}})^{-1}$
and $\mathbf{J}_{\mathbf{x}}$, $\mathbf{J}_{\mathbf{u}}$ have dimensions $2m \times 2m$ and $2m \times 2n$ respectively. 
Assuming no collisions between the entities,
in the herding problem $\mathbf{J}_{\mathbf{u}}^{+}$ is defined in $I$ because $h$ is of class $C^1$ in $I$. 

Evaders' and action dynamics together yields
\begin{equation}
\label{eq:diff_eq_system_u}
\left\{
\begin{aligned}
    \dot{\mathbf{x}} &= f(\mathbf{x}, \mathbf{u})
    \\ 
    \dot{\mathbf{u}} &= \mathbf{J}_{\mathbf{u}}^{T} \mathbf{J}_{\mathbf{u}}^{+}
    \kern -3pt\left(h^*(\mathbf{x},\mathbf{u})-\mathbf{J}_{\mathbf{x}} f(\mathbf{x}, \mathbf{u}) \right)
\end{aligned}
\right. ,
\end{equation}
which is an explicit system expanded from the original~\eqref{eq:initial_system}. 

In this problem, under no more requirements, we consider
\begin{equation}
\label{Eq:f_star_real}
f^*(\mathbf{x})=-\mathbf{K}_f\widetilde{\mathbf{x}},
\end{equation}
with $\mathbf{K}_f$ a positive definite matrix, as the desired closed-loop behaviour. This expression meets the conditions of continuity and differentiability of Theorem~\ref{theorem:exist_and_smooth} for all $\mathbf{x}$. On the other hand, we consider
\begin{equation}\label{eq:h_star}
    h^*(\mathbf{x},\mathbf{u}) = - \mathbf{K}_h h(\mathbf{x},\mathbf{u}),
\end{equation}
with $\mathbf{K}_h$ a positive definite matrix, so $\mathbf{0}$ is a globally asymptotically stable point of $dh/dt$. 

\begin{theorem}
\label{Theorem:u_dynamics}
Assume that Corollary~\ref{corollary:implicit_theorem} holds. Given $f^*$ in~\eqref{Eq:f_star_real} and $h^*$ in~\eqref{eq:h_star} such that the matrix
\begin{equation}\label{eq:negative_matrix}
    \mathbf{K} = \left.\begin{pmatrix}
    -\mathbf{K}_f & 0.5 \mathbf{I}_{2m}
    \\
    0.5 \mathbf{I}_{2m} & -\mathbf{K}_h
    \end{pmatrix}\right.
\end{equation}
is negative definite, then the system in~\eqref{eq:diff_eq_system_u} is globally asymptotically stable (GAS) in $I$.
\end{theorem}
\begin{proof}
We first show that~\eqref{eq:diff_eq_system_u} makes $dh/dt$ evolve according to $h^*$. We omit the dependencies with $\mathbf{x}$ and $\mathbf{u}$ to simplify the notation. Application of the chain rule to $h$ gives
\begin{equation}\label{eq:der_h}
    \frac{d h}{d t} = \frac{\partial h}{\partial \mathbf{x}} \frac{d \mathbf{x}}{d t} + \frac{\partial h}{\partial \mathbf{u}} \frac{d \mathbf{u}}{d t} = \mathbf{J}_{\mathbf{x}} \frac{d \mathbf{x}}{d t} + \mathbf{J}_{\mathbf{u}} \frac{d \mathbf{u}}{d t}.
\end{equation}
The substitution of~\eqref{eq:u_dynamics} in the right side of~\eqref{eq:der_h} yields to~\eqref{eq:h_star}, demonstrating that the action dynamics in~\eqref{eq:u_dynamics} makes $dh/dt$ evolve according to $h^*$. 

Now, let $V = \frac{1}{2}\widetilde{\mathbf{x}}^T\widetilde{\mathbf{x}} + \frac{1}{2}h^T h$ a Lyapunov function candidate. Its derivative is
\begin{equation}\label{eq:Ly_theorem_der}
    \dot{V} = \widetilde{\mathbf{x}}^T f + h^T h^* = \widetilde{\mathbf{x}}^T h + \widetilde{\mathbf{x}}^T f^* + h^T h^*.
\end{equation}
Substituting $f^*$ and $h^*$ in~\eqref{eq:Ly_theorem_der} gives
\begin{equation}\label{eq:Ly_theorem_der_2}
    \dot{V} = \widetilde{\mathbf{x}}^T h - \widetilde{\mathbf{x}}^T \mathbf{K}_f \widetilde{\mathbf{x}} - h^T \mathbf{K}_h h = \begin{pmatrix} \widetilde{\mathbf{x}}^T\! & \! h^T \end{pmatrix} \mathbf{K} \begin{pmatrix} \widetilde{\mathbf{x}}^T \!& \! h^T \end{pmatrix}^T
\end{equation}

If $\mathbf{K}_f$ and $\mathbf{K}_h$ are designed such that $\mathbf{K}$ is negative definite, then the system defined in~\eqref{eq:diff_eq_system_u} is GAS, driving both $h$ and $\widetilde{\mathbf{x}}$ to zero.
\end{proof}

In practice, by choosing $||\mathbf{K}_h|| \gg ||\mathbf{K}_f||$ we can impose the convergence of $h$ to its roots to be much faster than the desired closed-loop dynamics, so the evaders behave following $f^*$.

As a corollary of the previous results, the Explicit Design can deal with discontinuous dynamics in the position of the evaders like the Exponential Model in~\eqref{eq:LicitraBase} due to the fact that we are designing over $\dot{\mathbf{u}}$ instead of $\mathbf{u}$. This is validated in the simulations and experiments of Sections~\ref{sec:simulations} and~\ref{sec:experiments}.

The concept behind the Explicit Design can be observed in the literature, as in~\cite{Blanchini2017ModelFree}, where the time derivative of the action is used to characterise and tune static plants. However, to the best of our knowledge, it has never been applied in dynamical systems to build an explicit controller. 

From an algorithmic point of view, the calculation of the control action is very simple. At each instant, the controller receives $\mathbf{x}$ and $\mathbf{u}$ from an observer and/or from measurements. Then, we compute $f^*(\mathbf{x})$ with~\eqref{Eq:f_star_real}, which, together with the dynamic model of the evaders $f(\mathbf{x},\mathbf{u})$ (e.g.,~\eqref{eq:PiersonBase} or~\eqref{eq:LicitraBase}), allow us to compute $h^*(\mathbf{x},\mathbf{u})$ with Eq.~\eqref{eq:h_star}. Besides, we calculate the Jacobians $\mathbf{J}_{\mathbf{x}},\mathbf{J}_{\mathbf{u}}$ either analytically or numerically, depending on the complexity of finding their closed expression. Finally, we compute $\dot{\mathbf{u}}$ from Eq.~\eqref{eq:diff_eq_system_u}. 

\subsection{Implicit Design}\label{subsec:implicit_design}

A conceptually simpler alternative is to compute the action using a numerical method, giving rise to a discrete solution. Under the compliance of Corollary~\ref{corollary:implicit_theorem}, $\mathbf{u}$ exists and is smooth, so we can find the roots of $h$ with a standard numerical method to impose the desired dynamics $f^*$ in the evaders, always under the particular conditions of the numerical method. Despite the simpler approach, the dependency on the numerical method restricts the control strategy to situations where the configuration of the herders is, in general, closer to the roots of $h$ than with the Explicit Design.

At each instant, the numerical method receives $\mathbf{x}$ and $\mathbf{u}$. Then, in each iteration of the numerical method, denoted by the index $k$,  $h(\mathbf{x},\mathbf{u}_{k-1})$ and $\mathbf{J}_{\mathbf{u}}^{k-1}$ are calculated in order to obtain $\mathbf{u}_{k}$. In this work we use \textit{Levenberg-Marquardt}~\cite{Marquardt1963LM} as numerical method, yielding to the following iteration
\begin{equation}\label{eq:LM}
\begin{aligned}
    \mathbf{u}_{k} =& \mathbf{u}_{k-1} + \mathbf{\zeta} = 
    \\
    &\mathbf{u}_{k-1} - ((\mathbf{J}_{\mathbf{u}}^{k-1})^T \mathbf{J}_{\mathbf{u}}^{k-1} + \lambda \mathbf{I}_{2n})^{-1} (\mathbf{J}_{\mathbf{u}}^{k-1})^{T} h(\mathbf{x},\mathbf{u}_{k-1}),
\end{aligned}    
\end{equation}
where $\lambda$ is a weighting factor. If $||\mathbf{\zeta}||$, where $||\cdot||$ denotes the Euclidean norm, is less than a tolerance $\epsilon$ before reaching $k_{max}$, the method stops. 

From a more practical point of view, the solution is implemented in discrete time with sample time $T$. The sample time needs to be sufficiently small to maintain stability, but also large enough to ensure convergence to the roots of $h$. This implies a trade-off in the election of $T$, which will be assessed in  Sections~\ref{sec:simulations} and~\ref{sec:experiments}. The values of $\epsilon$ and $k_{max}$ are adjusted to fit the accuracy in the motion and computing capabilities of the robots, whereas $\lambda$ achieves smoothness in the movements of the robots.

As a last comment, the Implicit Control formulation allows us to deal, again, with discontinuous dynamics in the position of the evaders like the Exponential Model in~\eqref{eq:LicitraBase}. Since, at each instant, the numerical method performs a series of iterations departing from the previous computed action, the only effect of the discontinuity will be that of yielding to more iterations to reach the roots of $h$ in the very first instants after the discontinuity. 


\section{Simulation results}\label{sec:simulations}
This Section describes the results of testing the herding control strategy in simulations. The objectives with the simulations are twofold. 1) Validation and comparison between design methods. 2) Demonstration of the success of the proposal against challenging situations using the models in Section~\ref{sec:prosta}. Times have been calculated using \textit{tic-toc} functions of Matlab 2019b in a laptop with an Intel Core i7-5500U CPS at 2.4GHz. In the simulations, we focus on the performance of our proposal in the herding, so the approaching and assignment are supposed to be previously executed and completed. 
The first case of study consists in the herding of 5 evaders by 5 herders, and the details are in Table~\ref{table:initial_values}.
\begin{figure*}[!ht]
    \centering
    \begin{tabular}{cccc}
         {\footnotesize (a)}
         &  
         {\footnotesize (b)}
         &
         {\footnotesize (c)}
         &
         {\footnotesize (d)}
         \\\hspace{-0.6cm} \hspace{0.1cm}
         \includegraphics[width=0.21\textwidth]{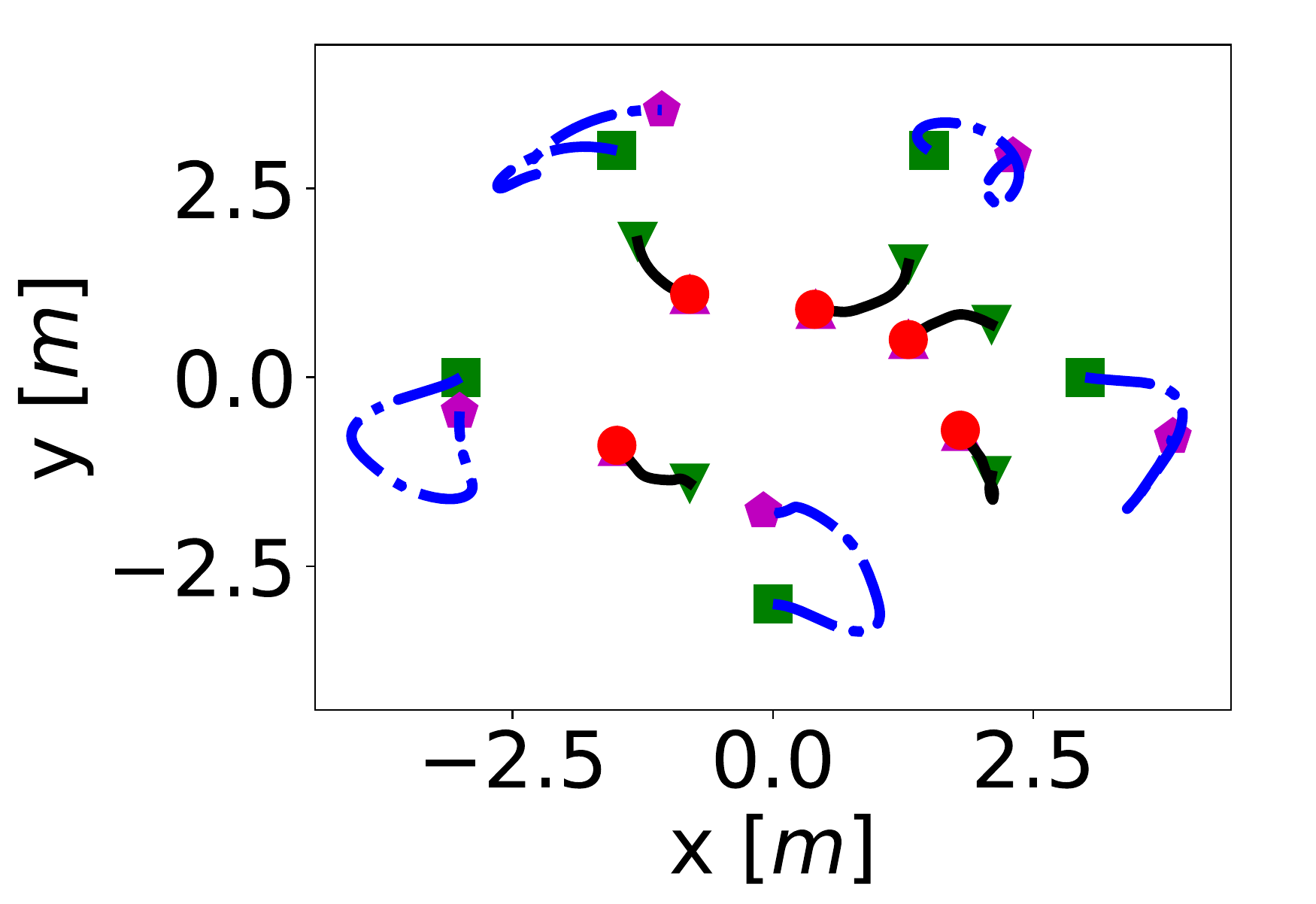}\hspace{0.1cm}
         &\hspace{-0.45cm}\hspace{0.1cm}
         \includegraphics[width=0.21\textwidth]{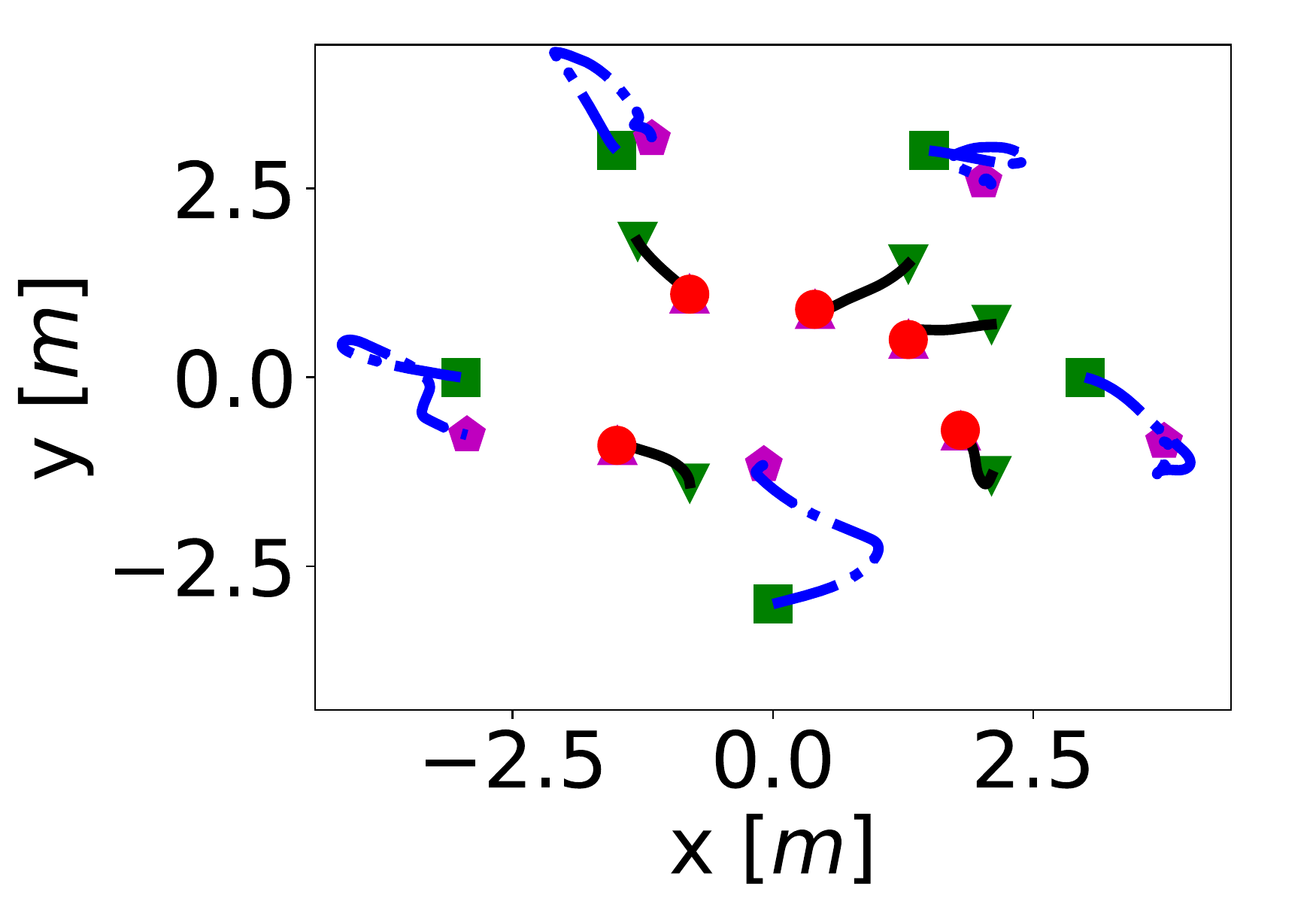}\hspace{0.1cm}
         &\hspace{-0.4cm}\hspace{0.1cm}
         \includegraphics[width=0.21\textwidth]{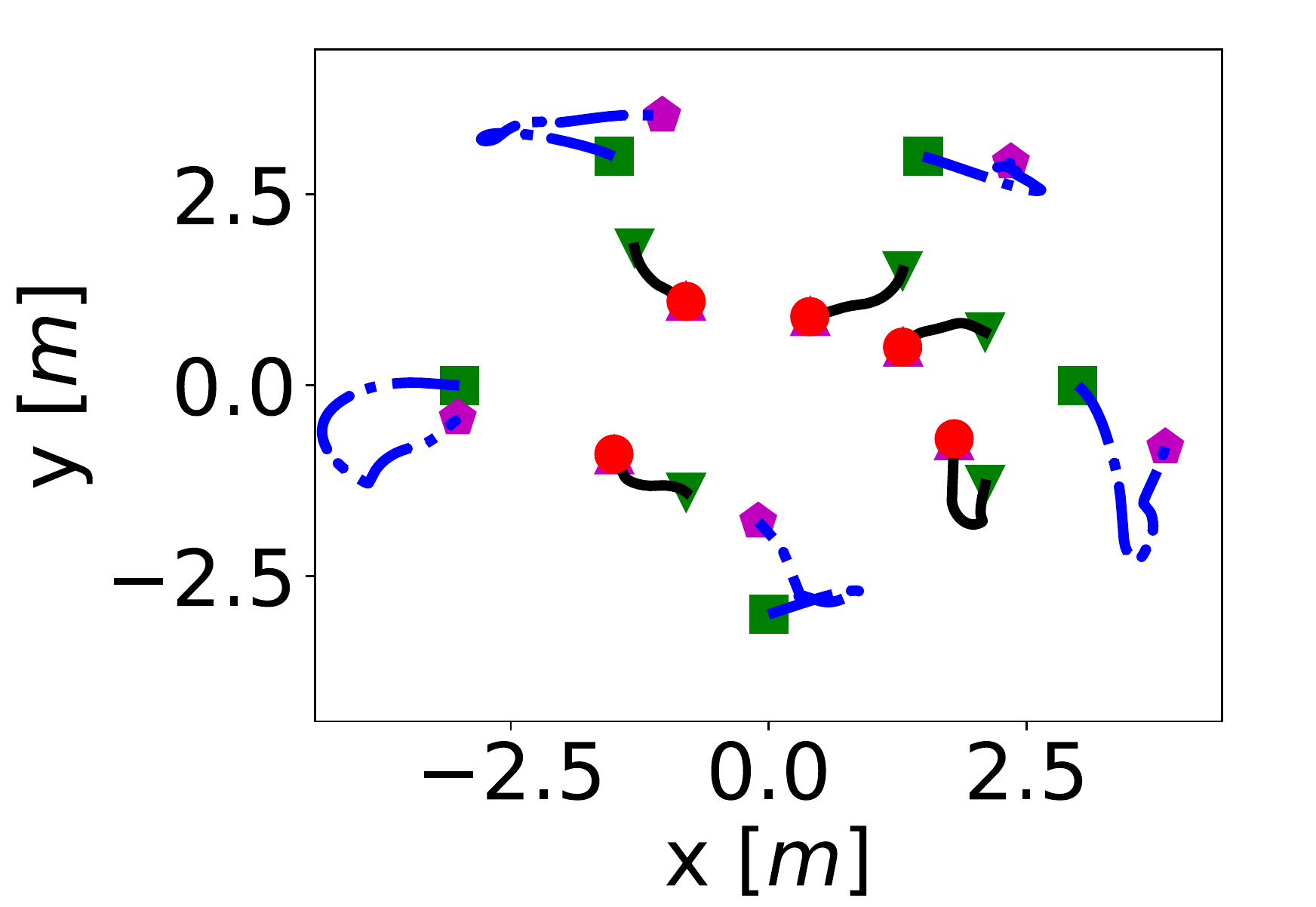}\hspace{0.1cm}
         &\hspace{-0.4cm}\hspace{0.1cm}
         \includegraphics[width=0.21\textwidth]{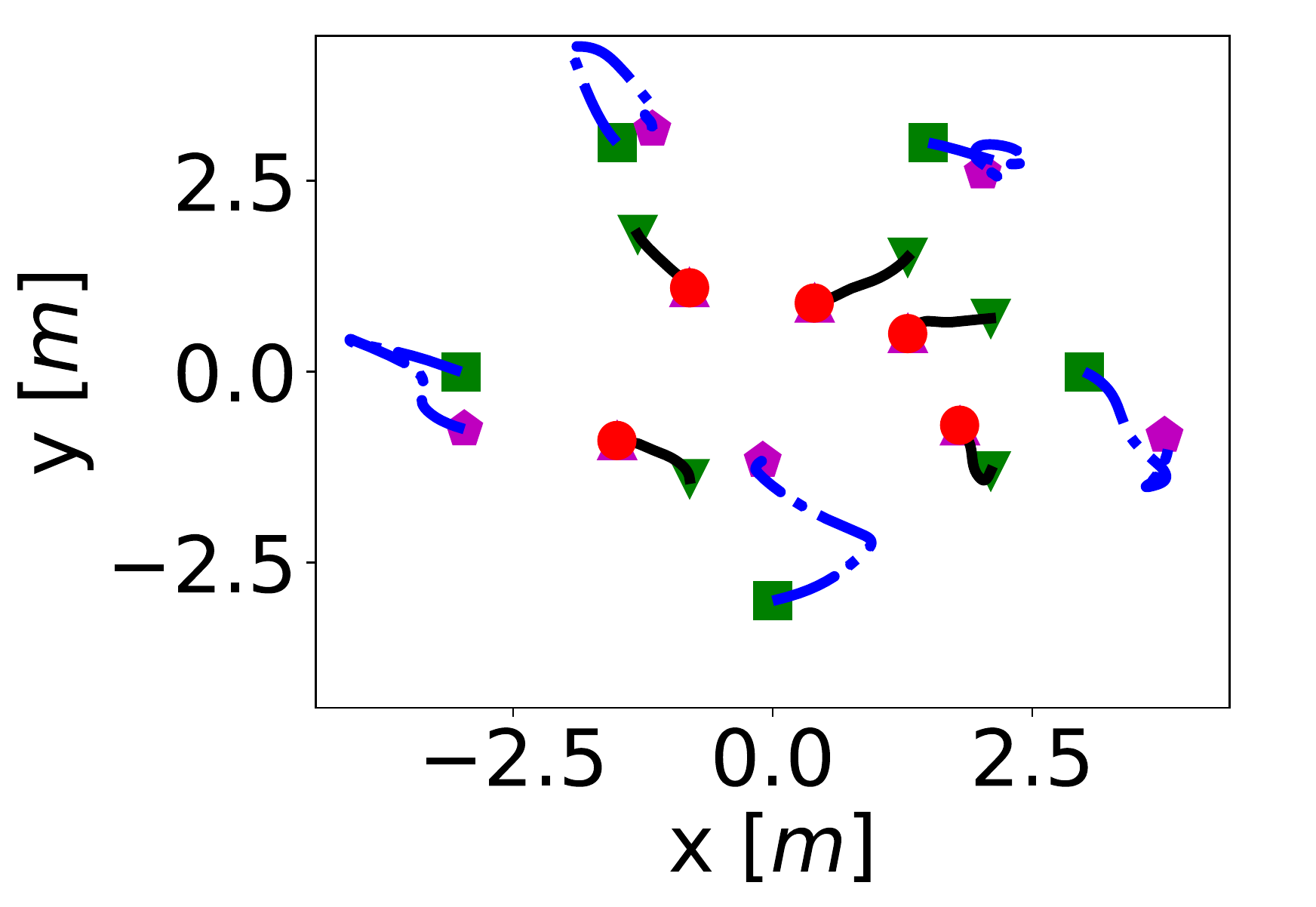}\hspace{0.1cm}
         \\\hspace{0.1cm}
         \includegraphics[width=0.21\textwidth]{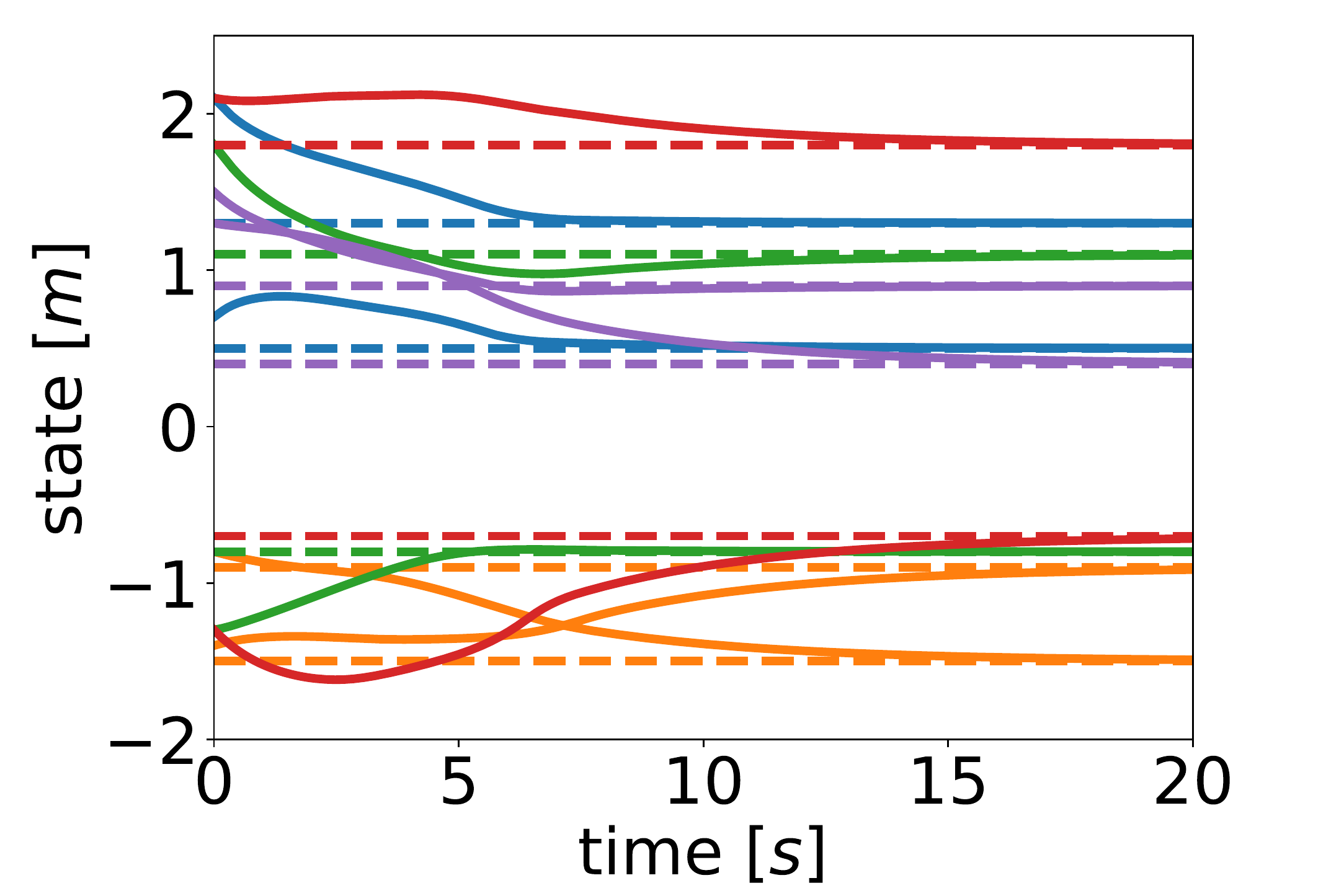}\hspace{0.1cm}
         &\hspace{0.1cm}
         \includegraphics[width=0.21\textwidth]{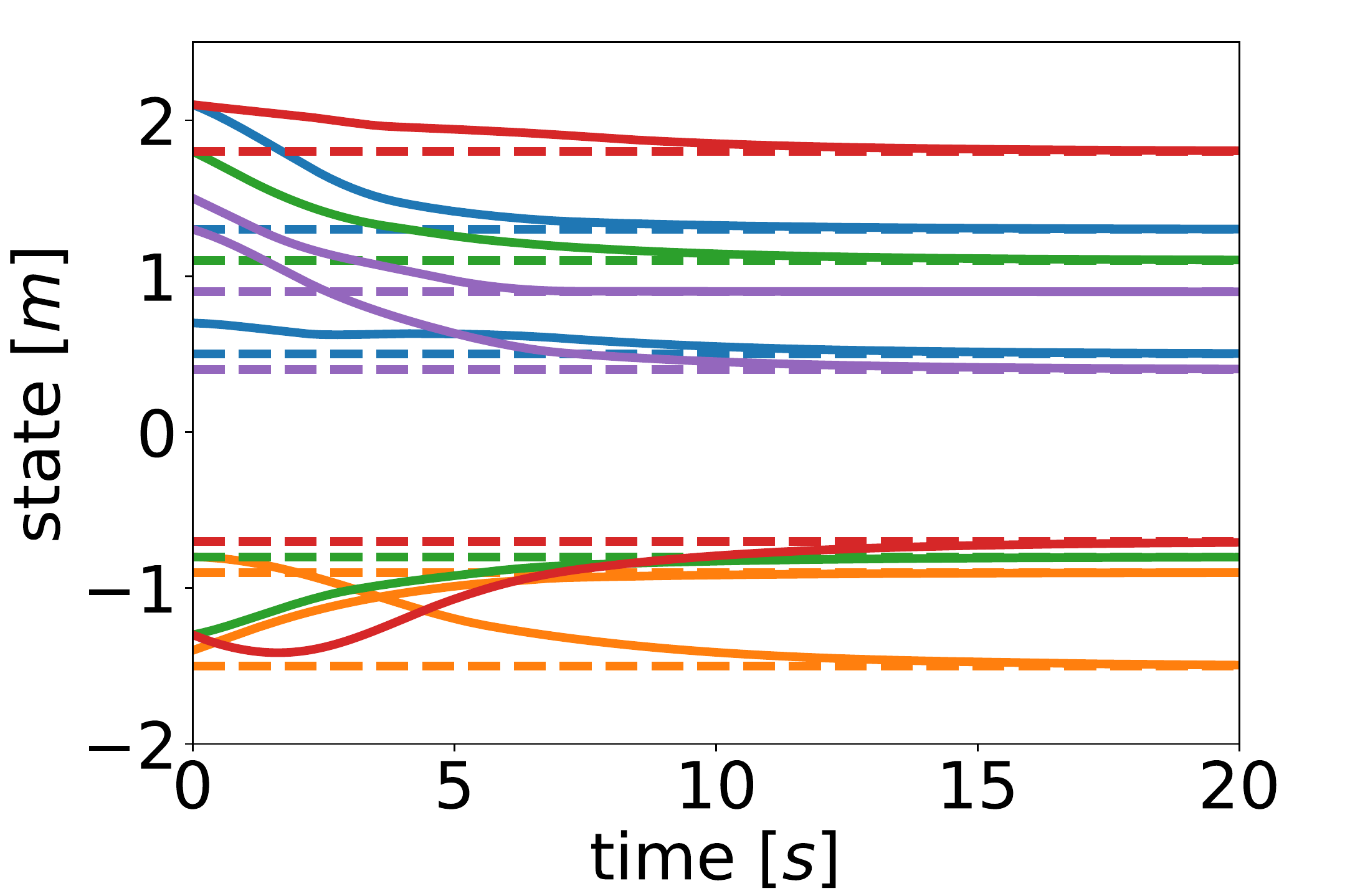}\hspace{0.1cm}
         &\hspace{0.1cm}
         \includegraphics[width=0.21\textwidth]{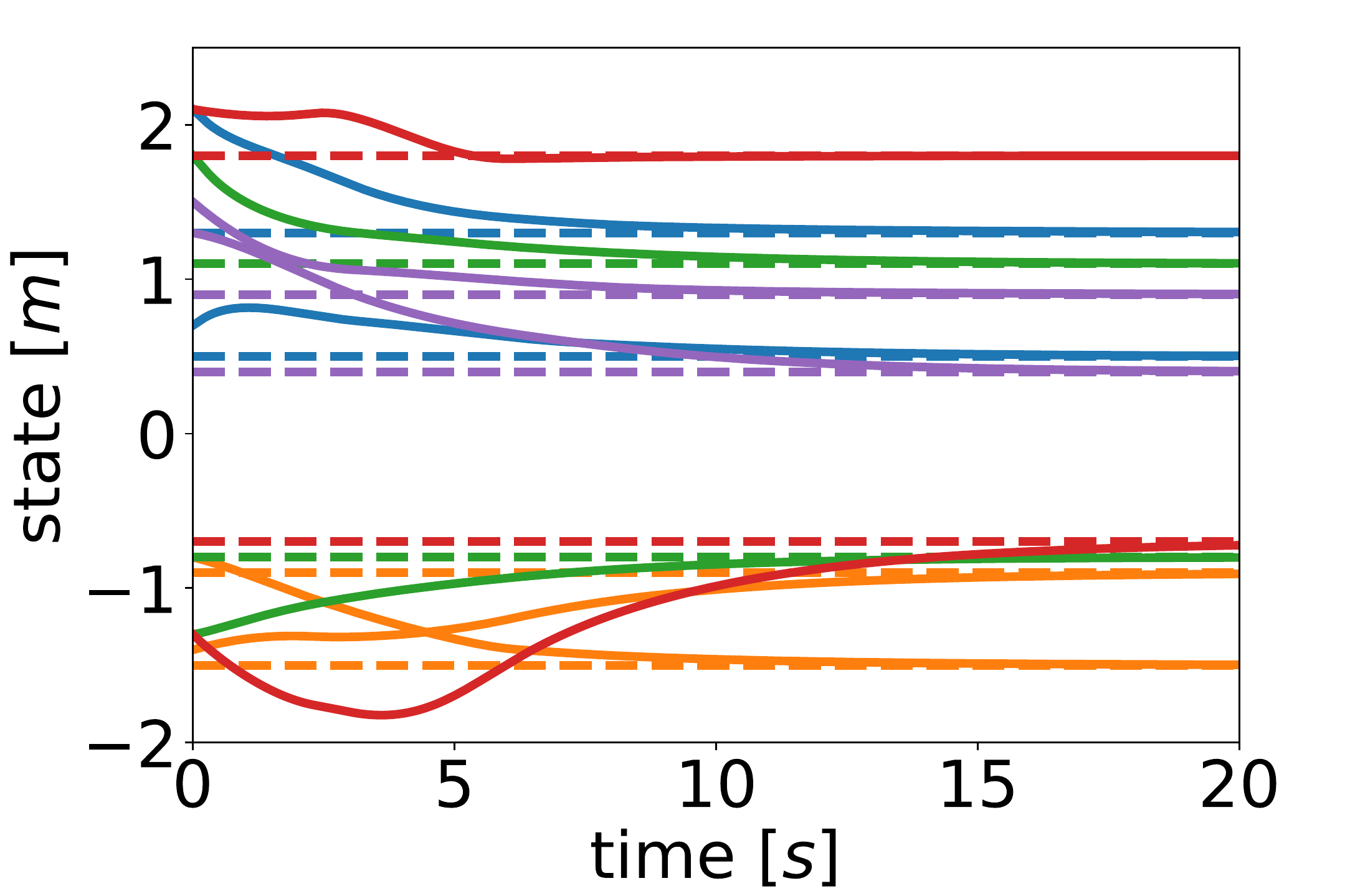}\hspace{0.1cm}
         &\hspace{0.1cm}
         \includegraphics[width=0.21\textwidth]{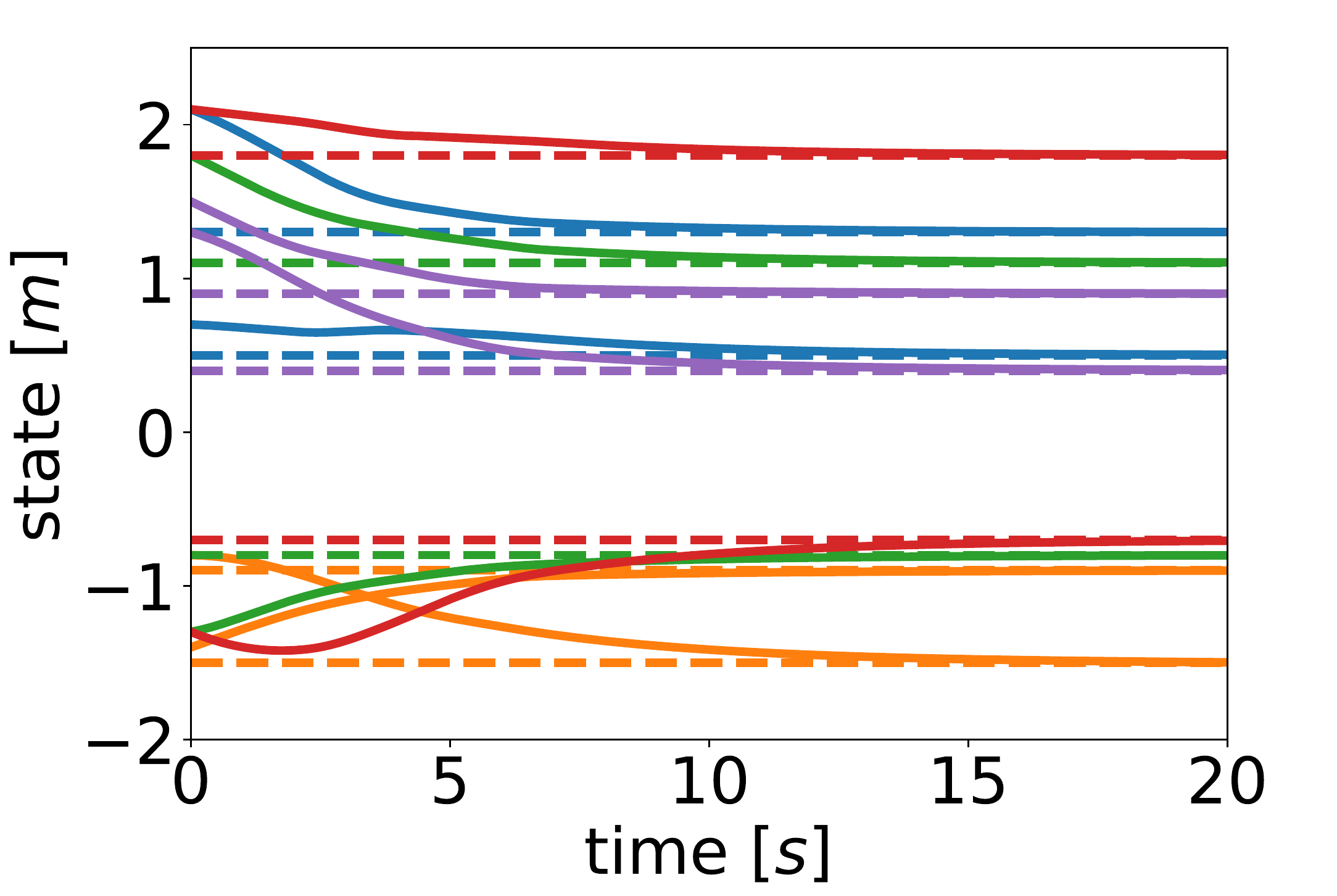}\hspace{0.1cm}
    \end{tabular}
	\caption{Simulation results of the herding of 5 evaders by 5 herders: (a) Inverse evaders, Explicit solution, (b) Exponential evaders, Explicit solution, (c) Inverse evaders, Implicit solution, (d) Exponential evaders, Implicit solution. The first row shows the trajectories followed by the herders (blue-dashed lines) and the evaders (black lines), departing from the green squares, ending in the magenta pentagons and with desired positions in red circles. The second row presents the evolution of the position of the evaders with time, where the dashed lines mark the desired reference.}
	\label{fig:sim_results}
\end{figure*}
\begin{table}[!ht]
\centering
\caption{Initial configuration and parameters of the simulation.}
\begin{tabular}{c}
\begin{tabular}{c|c|c|c|c|c|c|c|c|c|c}
 &$x_1$&$y_1$&$x_2$&$y_2$&$x_3$&$y_3$&$x_4$&$y_4$&$x_5$& $y_5$\\
\hline
$\mathbf{x}$[m]& $2.1$ & $0.7$ & $-0.8$ & $-1.4$ & $-1.3$ & $1.8$ & $2.1$ & $-1.3$ & $1.3$ & $1.5$\\
\hline
$\mathbf{x}^*$[m]& $1.3$ & $0.5$ & $-1.5$ & $-0.9$ & $-0.8$ & $1.1$ & $1.8$ & $-0.7$ & $0.4$ & $0.9$\\
\hline
$\mathbf{u}$[m]& $-3.0$ & $0.0$ & $-1.5$ & $3.0$ & $3.0$ & $0.0$ & $0.0$ & $-3.0$ & $1.5$ & $3.0$
\end{tabular}
\\
\vspace{-0.2cm}
\\
\begin{tabular}{c|c|c|c|c|c|c|c}
 $\gamma_j$&$\alpha_j$&$\beta_j$&$\sigma_j$&$\lambda$&$\epsilon$&$k_{max}$& $r$\\
\hline
$1.0$ $\forall j$&$0.6$ $\forall j$&$0.5$ $\forall j$ & $2.0$ $\forall j$ & $0.1$ & $1$mm & $20$ & $1.0$m
\end{tabular}
\end{tabular}
\label{table:initial_values}
\end{table}

For implementation purposes, a sample time $T = 10$ms is chosen and both evaders and herders have the same maximum velocity $v_{max} = 0.4$m/s.
We set $\mathbf{K}_f = 0.25 \mathbf{I}_{2 m}$, yielding a settling time of $12$s. To ensure that the conditions of Theorem~\ref{Theorem:u_dynamics} hold, we set $\mathbf{K}_h = 50 \mathbf{I}_{2 m}$. 
Due to the complexity of the models and $h$, the Jacobians $\mathbf{J}_{\mathbf{u}}$ and $\mathbf{J}_{\mathbf{x}}$ are computed by numerical differentiation.

The first row of Fig.~\ref{fig:sim_results} shows the trajectories followed by herders and evaders for the different test cases.
Both methods are able to control the evaders successfully with similar behaviour in all the experiments.
Conversely, the performance of the herders presents some differences depending on the model, but they are almost equivalent among control algorithms.
This highlights the complexity of the control problem at hand, greater for the Exponential Model than the Inverse because of the switching dynamics.

The second row of Fig.~\ref{fig:sim_results} represents the evolution over time of the position of the evaders. When the Implicit Design is applied, the state evolves as an exponential function, reaching the desired position in $12$s according to the imposed closed-loop dynamics. On the other hand, the Explicit Design achieves the desired settling time with a slightly different transient response. 
The reason for that can be seen in Fig.~\ref{fig:diff_results}, where we show the difference between the Explicit and Implicit action. Initially, $dh/dt$ has not converged to zero. This convergence is subject to the numerical method in the Implicit Design, so performing enough iterations makes $dh/dt$ go to zero in the next instants. In the Explicit Design, the convergence depends on $h^*$, which in general is slower than the Implicit Design. Once $dh/dt$ has converged to zero, the trajectories follow the ideal form and the difference between actions vanishes.

\begin{figure}[!ht]
    \centering
    \begin{tabular}{cc}
        \includegraphics[width=0.46\columnwidth,height=0.30\columnwidth]{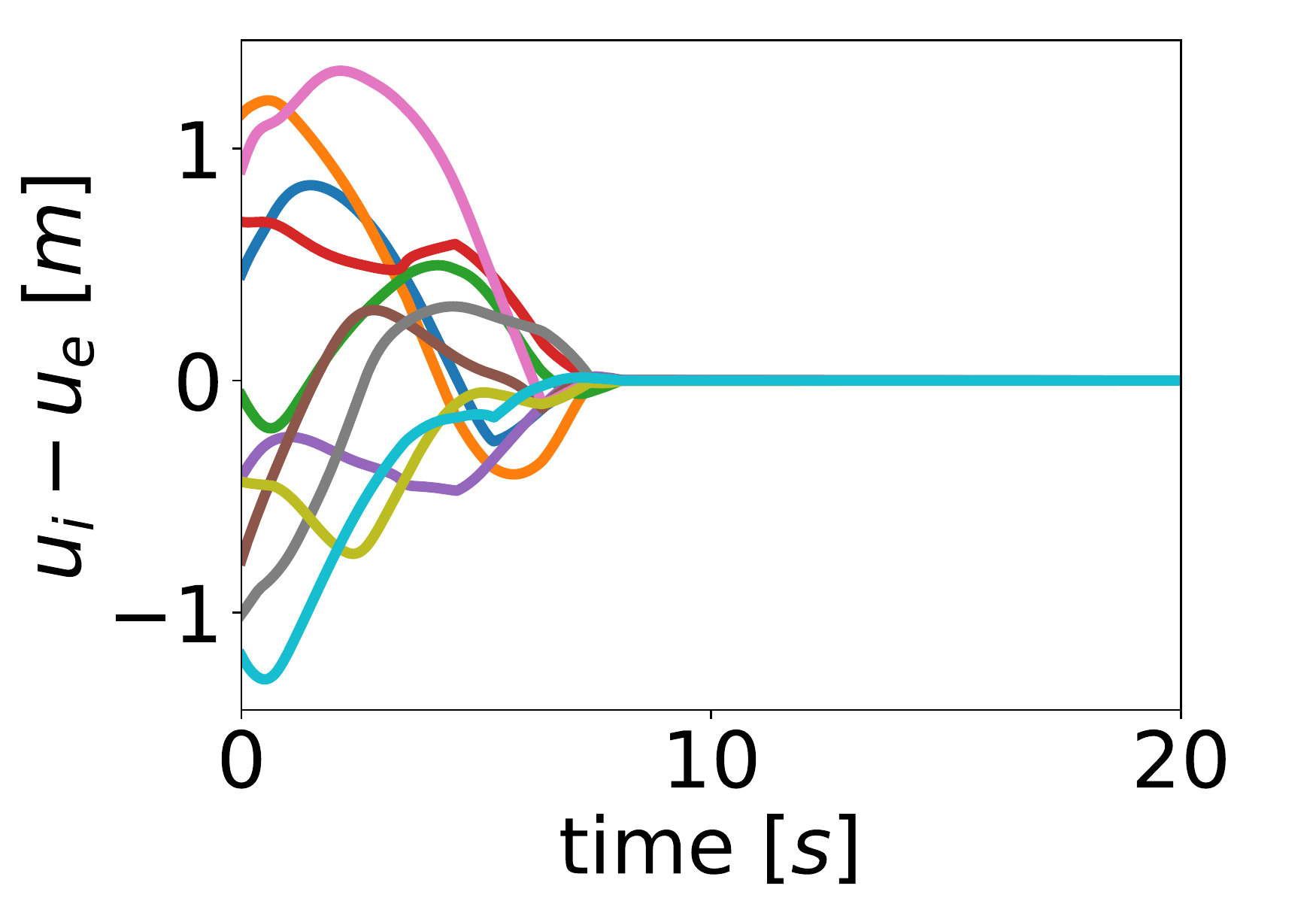}
         &  
        \includegraphics[width=0.46\columnwidth,height=0.30\columnwidth]{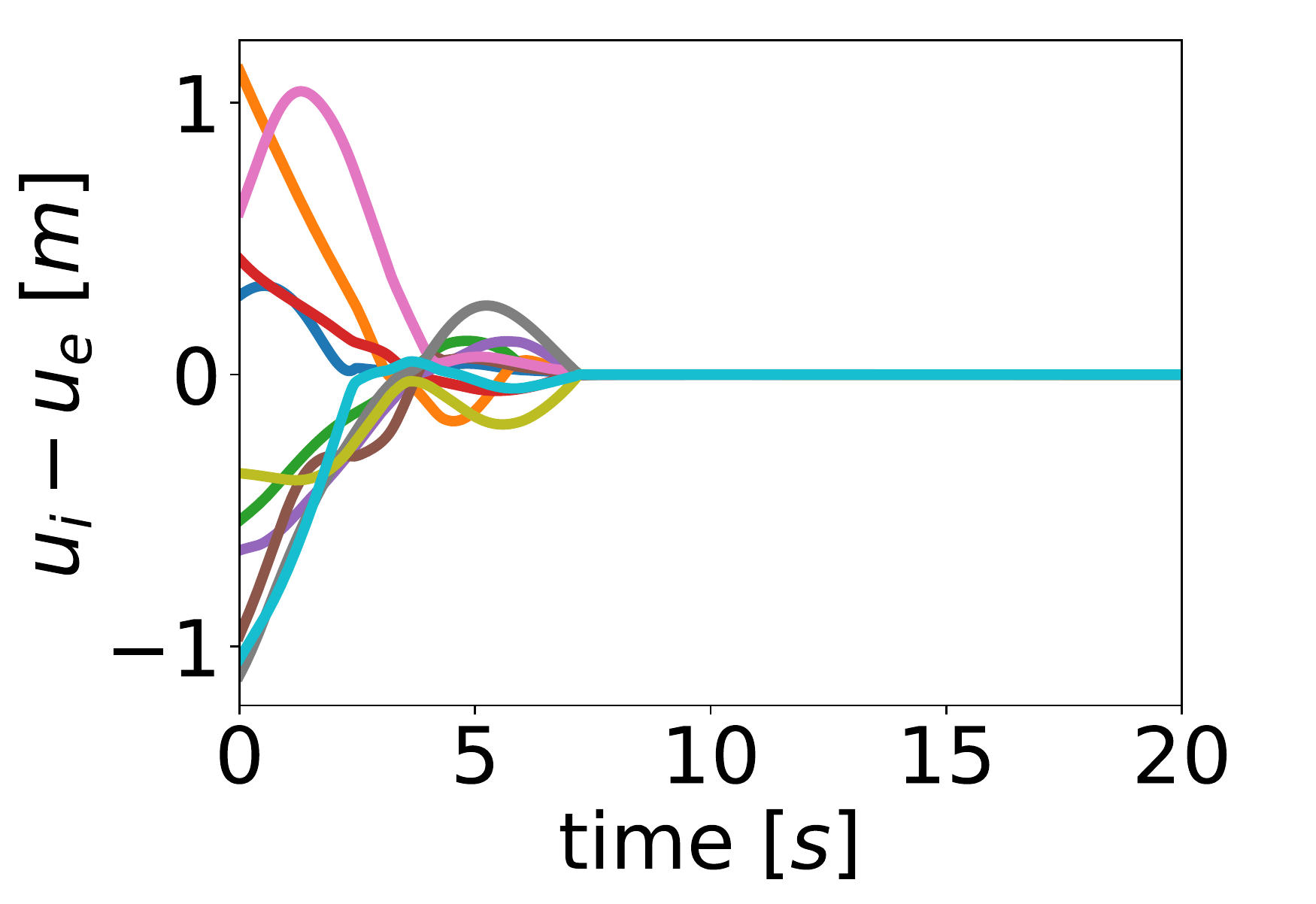}
    \end{tabular}
	\caption{Difference in the control action calculated by the Implicit Design and the Explicit Design: (left) Inverse Model, (right) Exponential Model. An arbitrarily colour has been assigned to each action in the sake of visibility.}
	\label{fig:diff_results}
\end{figure}

The Explicit solution is computationally cheaper since no iterations are needed to calculate the action. Nevertheless, the Implicit procedure is more reliable in imposing the desired closed-loop dynamics and it is easier to pose, so it is interesting to study the influence of the numerical method in the success of the solution.

In particular, we study the influence of the sampling time $T$ in the performance and computational cost of the Implicit Design. In Table~\ref{table:costeComp} we test the same examples of Fig.~\ref{fig:sim_results} by changing $T$, with values of $0.01$s (the original), $0.1$s and $0.5$s, removing the maximum number of iterations $k_{max}$ to allow the numerical method to run until the roots are found. In Table~\ref{table:costeComp}, $\tau$ is the computing time invested in one sample, whilst $\eta= ||h(\mathbf{x},\mathbf{u})||$ is the residual of the numerical method. The symbols $\overline{\cdot}$ and $\sigma(\cdot)$ denote the mean and standard deviation. 

A greater $T$ implies more time to compute the action but also a greater change in $h$, so the numerical method needs more iterations to converge. These conclusions are corroborated in Table~\ref{table:costeComp}, showing that the performance of the numerical method is enhanced when $T$ is small, achieving smaller values of $\eta$ with fewer iterations. Indeed, the row marked in red in Table~\ref{table:costeComp}, namely, the Exponential Model case with $T=0.5$s, fails and exposes this trade-off.

\begin{table}[!ht]
\centering
\caption{Evaluation of the numerical method as a function of $T$.}
\begin{tabular}{|c|c|c|c|c|c|c|}
\hline
Model & $T$ (ms) & $\overline{\tau}$ (ms) & $\overline{k}$ & $\sigma (k)$ & $\overline{\eta}\cdot 10^{-3}$ & $\sigma(\eta) \cdot 10^{-3}$\\
\hline
Inv.  & $10$ & $1.8$ & $2.4$ & $5.4$ & $0.7$  & $0.6$\\
Inv.  & $100$ & $3.9$ & $8.2$ & $28.9$ & $0.9$ & $1.2$\\
Inv.  & $500$ & $28.1$ & $57.7$ & $144.9$ & $8.2$ & $25.7$\\ 
\hline
Exp.  & $10$ & $6.6$ & $15.9$ & $77.1$ & $23.5$ & $10.2$\\
Exp.  & $100$ & $89.3$ & $264.5$ & $413.2$ & $32.9$ & $23.9$\\
\rowcolor{LightCyan}
Exp.  & $500$ & $279.4$ & $930.3$ & $276.2$ & $1140.8$ & $633.5$\\ 
\hline
\end{tabular}
\label{table:costeComp}
\end{table}

The flexibility and generality of the solution can be extended to heterogeneous groups and time-varying references, resulting in a more realistic herding. The example in Fig.~\ref{fig:TV_example} shows how three herders herd a group of three evaders. The red evader is Exponential while the purple ones are Inverse. Both design procedures perform similarly, so here we only show the results of the Implicit Design.

The desired herding configuration evolves according to
\begin{align*}
    \dot{x}^*_j = v_j^* ,  
    & &
    \dot{y}^*_j = 0.5 w_j^* \cos(w_j^* t + 2 \pi / j)
\end{align*}
with $\mathbf{w}^* \!\!=\!\! [0.05, 0.1, 0.02]$rad/s and $\mathbf{v}^*\!\! =\!\! [0.05, 0.05, 0.05]$m/s.

Initially, the herders move to drive the evaders to their sine references. This yields to trajectories surrounding and modulating the interaction forces with the evaders. Once the evaders are in their desired trajectories, the system reaches a steady-state behaviour where the periodic movement of the evaders is shared by the herders. 
\begin{figure*}[!ht]
    \centering
    \begin{tabular}{ccccc}
        {\footnotesize (a) Global}
        &
        {\footnotesize (b) $t \in [0,12]$ s}
        &
        {\footnotesize (c) $t \in [12,172]$ s}
        &
        {\footnotesize (d) $t \in [172,332]$ s}
        &
        {\footnotesize (e) $t \in [332,492]$ s}
        \\
        \includegraphics[width=0.18\textwidth]{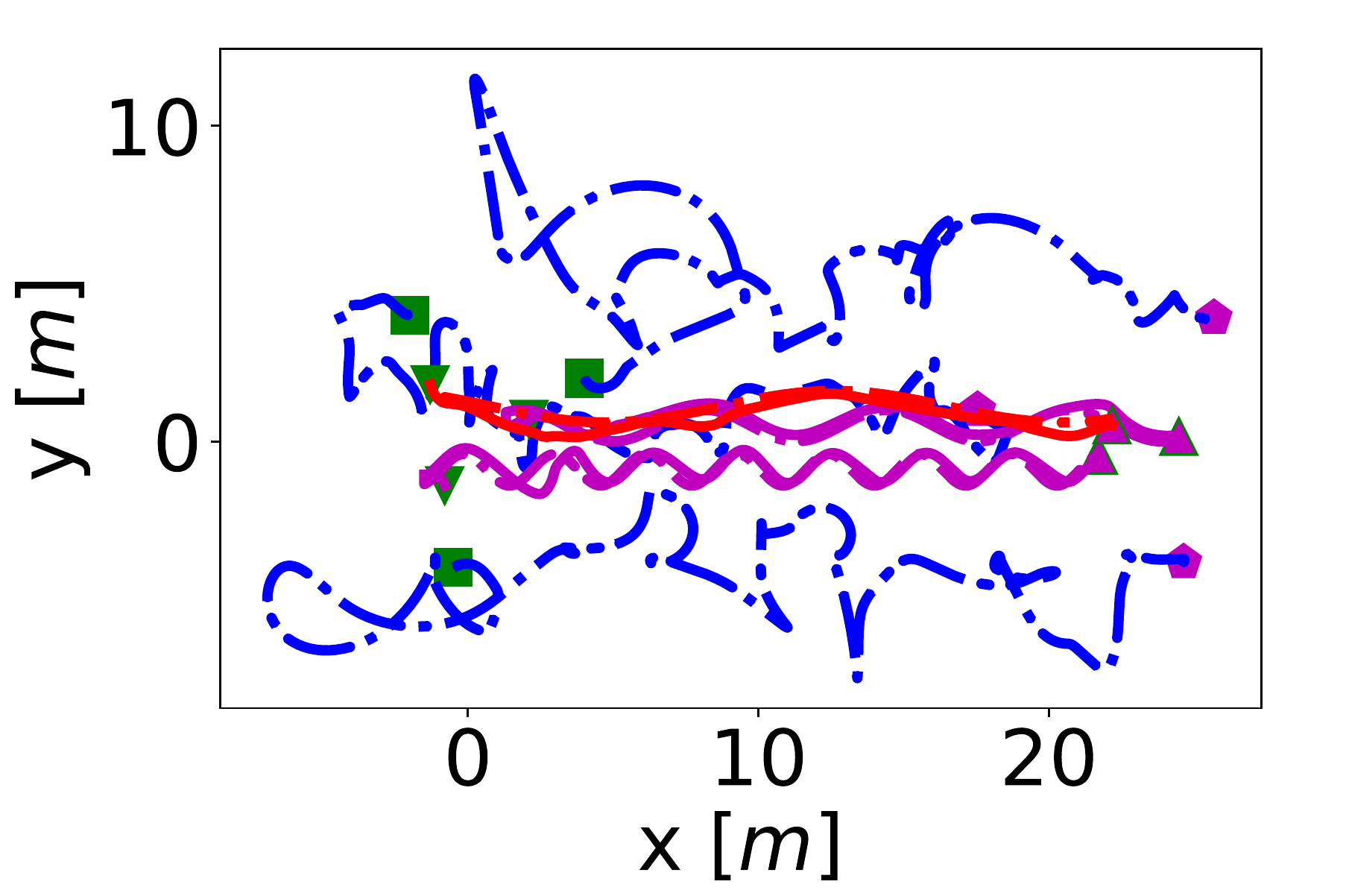}
        &\raisebox{.015\height}{ 
        \includegraphics[width=0.18\textwidth,height=0.123\textwidth]{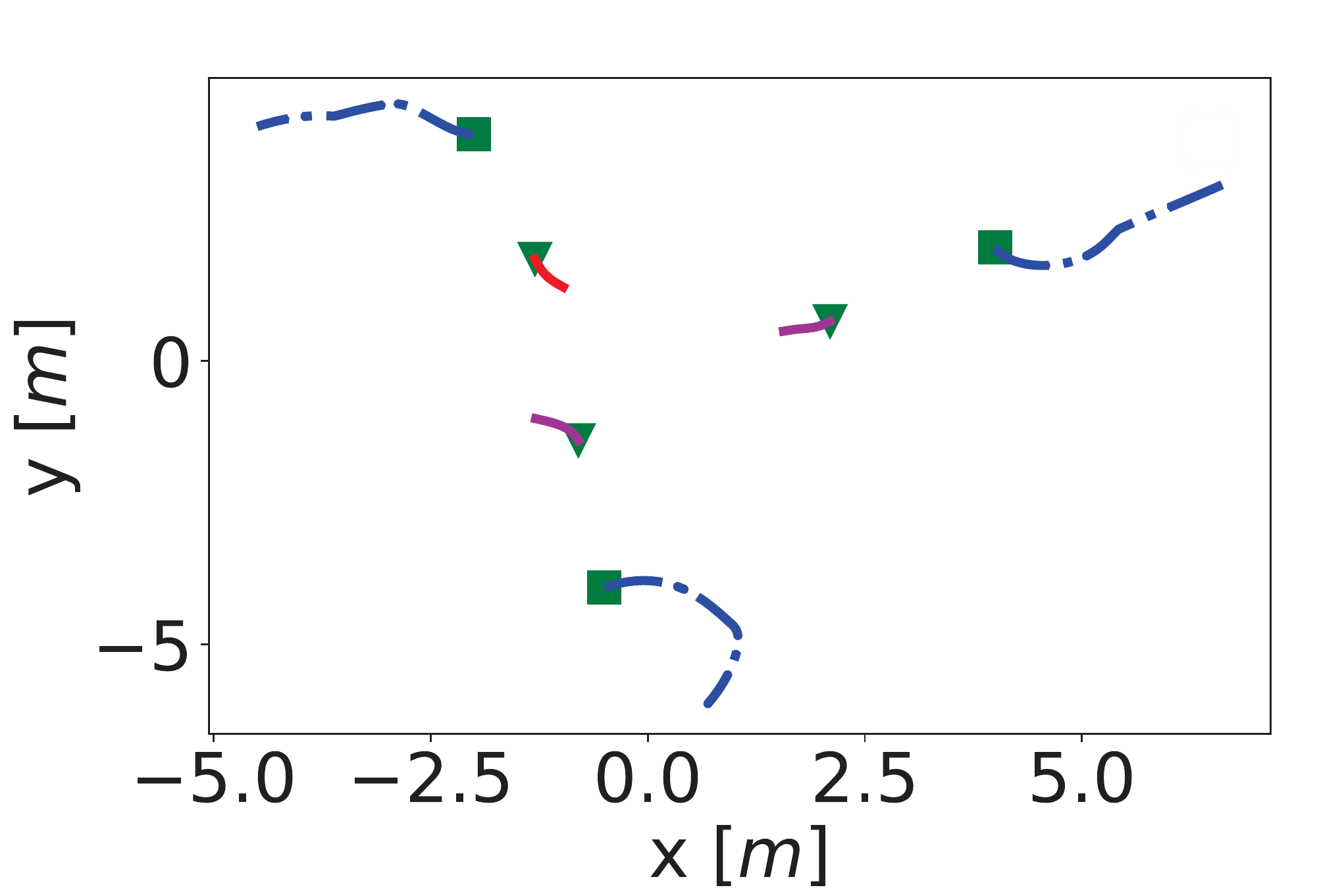}}
        &
        \includegraphics[width=0.18\textwidth,height=0.125\textwidth]{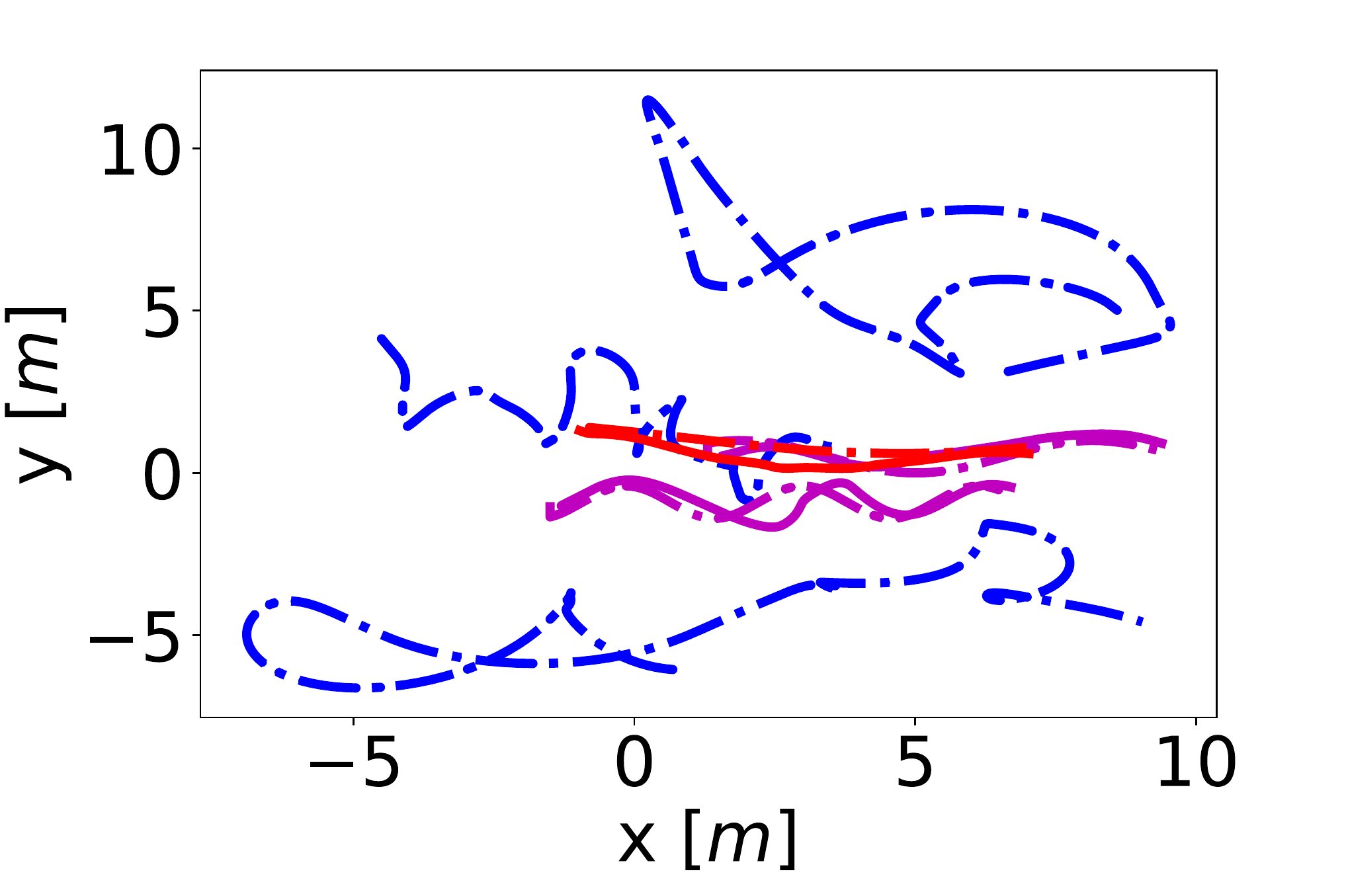}
        &\raisebox{-.02\height}{
        \includegraphics[width=0.18\textwidth,height=0.1265\textwidth]{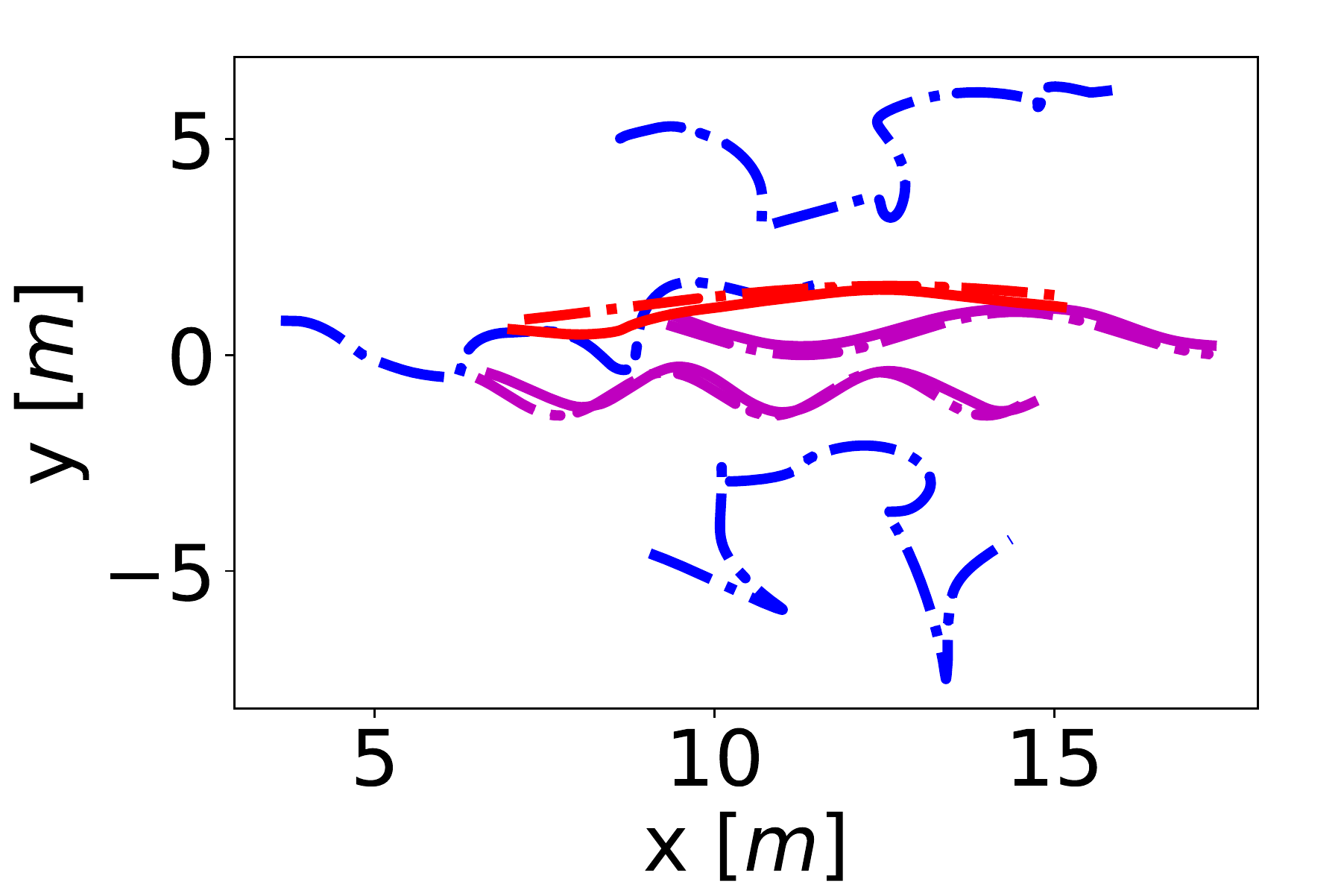}}
        &\raisebox{-.02\height}{
        \includegraphics[width=0.18\textwidth,height=0.1265\textwidth]{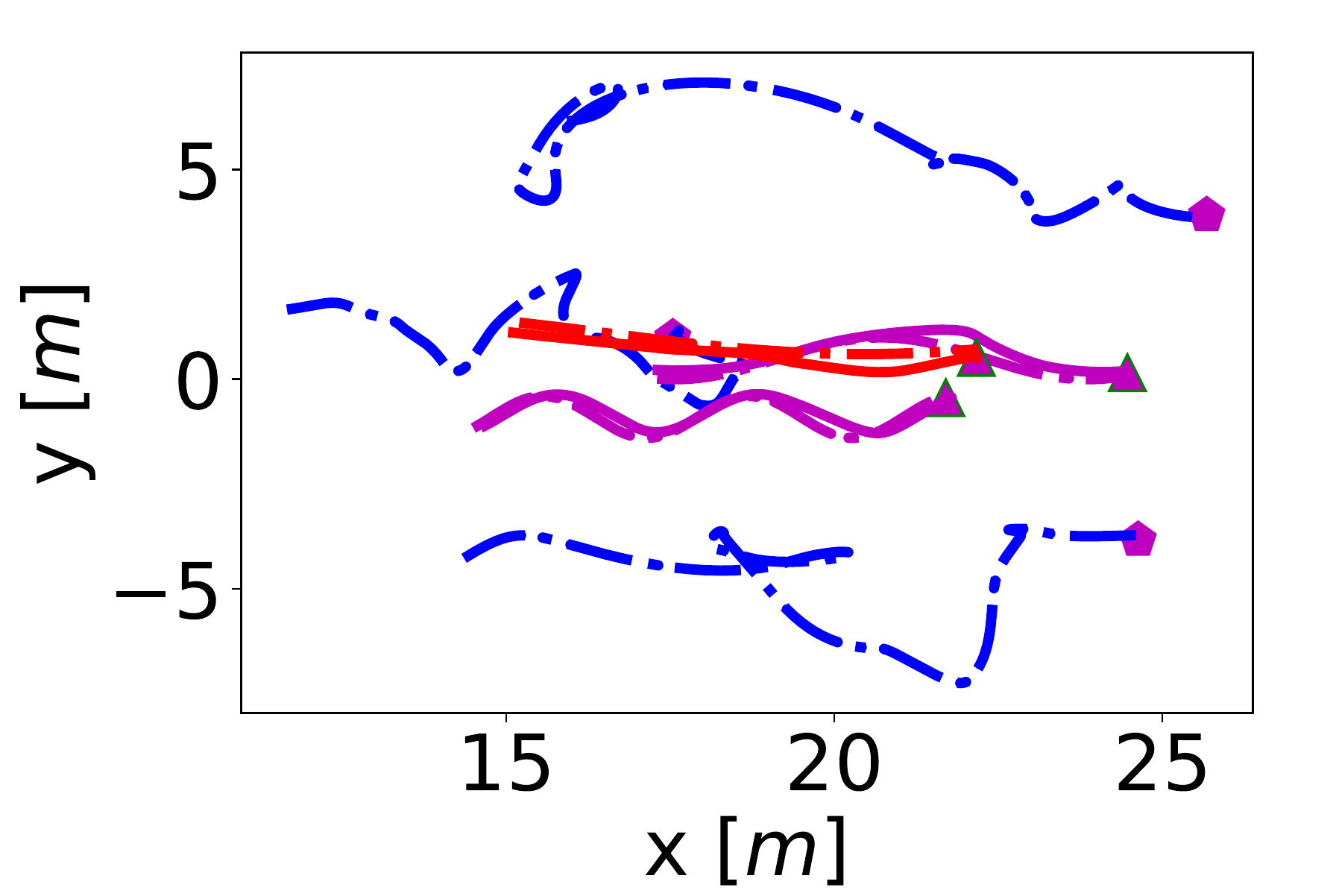}}
    \end{tabular}
	\caption{Three robotic herders herding three heterogeneous evaders: (a) complete trajectories followed by the entities, (b) first instants where herders drive evaders to initial location of the time-varying trajectories, (c)-(d) intermediate part of the herding, evaders following the time-varying trajectories, (e) evaders reaching the final desired positions. The purple evaders follow the Inverse Model and the red evader follows the Exponential Model. The other symbols and colours follow the convention in Fig.~\ref{fig:sim_results}.}
	\label{fig:TV_example}
\end{figure*}


\section{Experiments}\label{sec:experiments}
In this Section we extend the experiments to the real framework provided by the Robotarium arena~\cite{Pickem2017Robotarium}~\cite{Wilson2020Robotarium}. To do so, some robots play the role of herders while the others act as evaders, following the dynamics in Section~\ref{sec:prosta}. The robots are GRITSBot X playing in a $3.2$m x $2$m area, coordinated by a central server which receives odometry data and sends velocity commands to the robots at approximately delay of $0.033$s. Thus, a low level controller is used to translate the output of the control into velocity commands, with $v_{max}=0.2$m/s. Besides, robots use barrier certificates to avoid collisions. With this in mind, we adjust some parameters to fit the conditions of the experiment: $T=0.033$s, $\gamma = 0.025$, $\alpha=0.05$, $\sigma = 1.2$ for the static herding and $\mathbf{K}_f=0.5\mathbf{I}_{2m}$, $\gamma = 0.05$, $\alpha=0.2$ for the time-varying herding. 

\begin{figure*}[!ht]
\centering
\begin{tabular}{cccc}
    {\footnotesize (a)} & {\footnotesize (b)} & {\footnotesize (c)} & {\footnotesize (d)}
    \\\hspace{0.1cm}
    \includegraphics[width=0.21\textwidth]{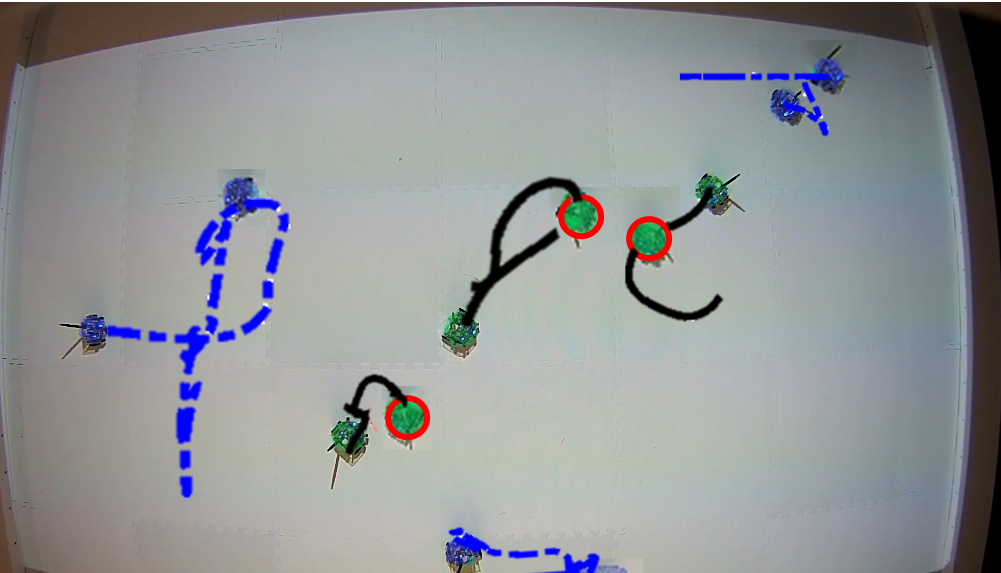}\hspace{0.1cm}
    & \hspace{0.1cm}
    \includegraphics[width=0.21\textwidth]{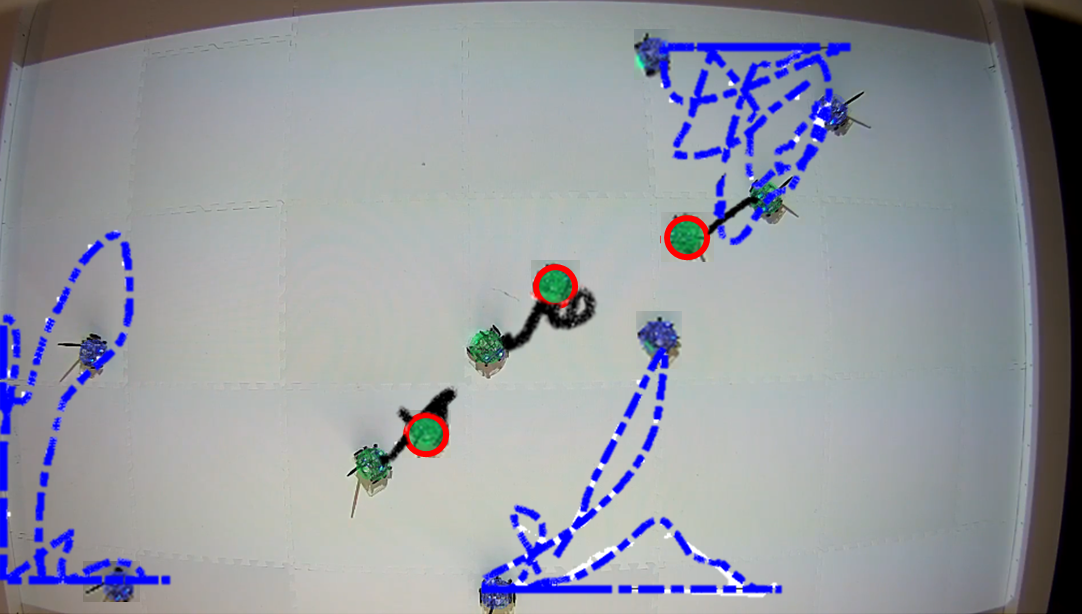}\hspace{0.1cm}
    &\hspace{0.1cm}
    \includegraphics[width=0.21\textwidth]{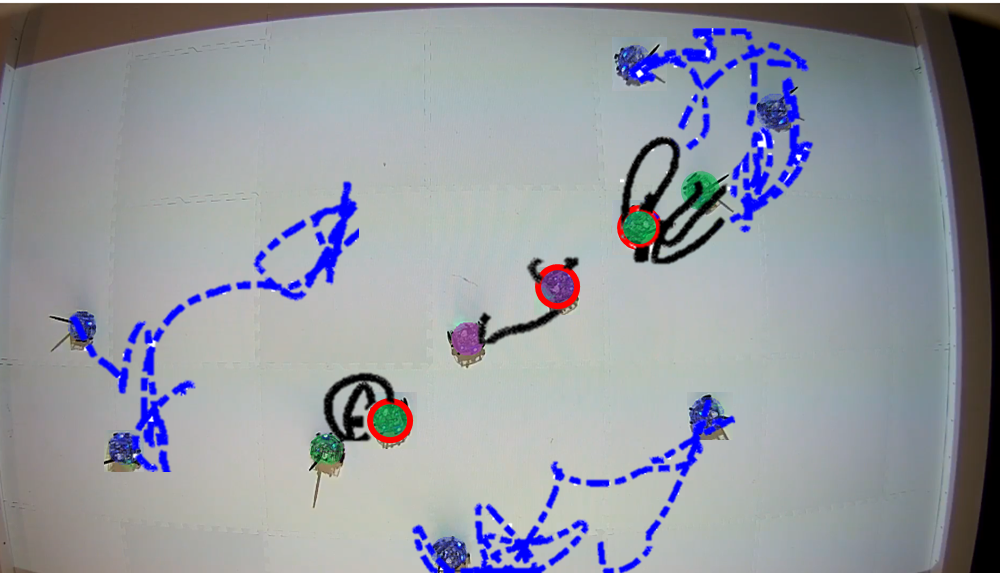}\hspace{0.1cm}
    &\hspace{0.1cm}
    \includegraphics[width=0.22\textwidth]{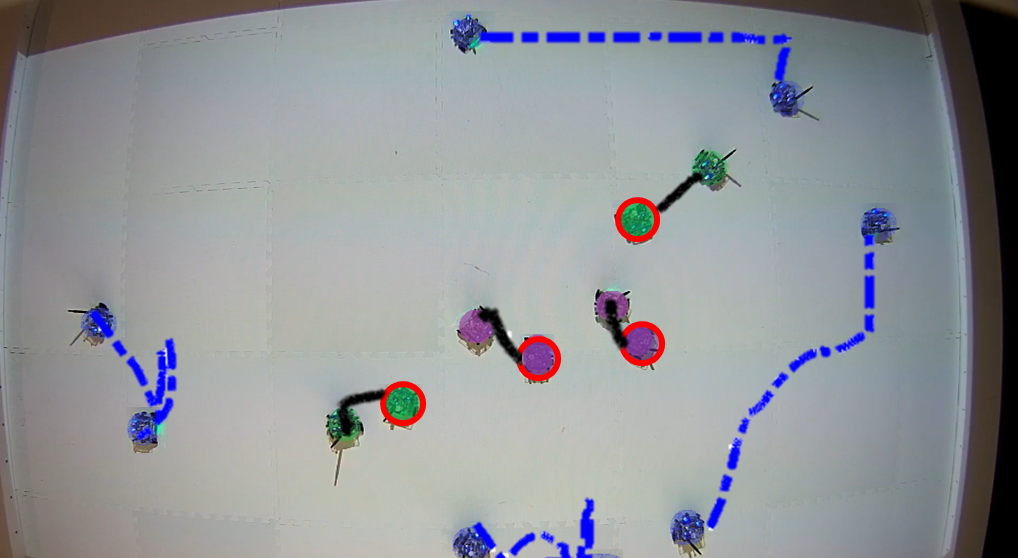}\hspace{0.1cm}
    \end{tabular}
\caption{Results of the herding using Robotarium~\cite{Pickem2017Robotarium}: (a) 3 herders vs 3 Inverse evaders, (b) 3 herders vs 3 Exponential evaders, (c) 3 herders vs 2 Inverse evaders and 1 Exponential evader, (d) 3 herders vs 2 Inverse evaders and 2 Exponential evaders. In green, the evaders (when mixed, Exponential evaders in purple); in blue, the herders; red contours denote desired positions. Video of the experiments is included as supplementary material.
}
\label{fig:second_impression}
\end{figure*}

\begin{figure}[!ht]
    \centering
    \begin{tabular}{cc}
        \includegraphics[width=0.43\columnwidth]{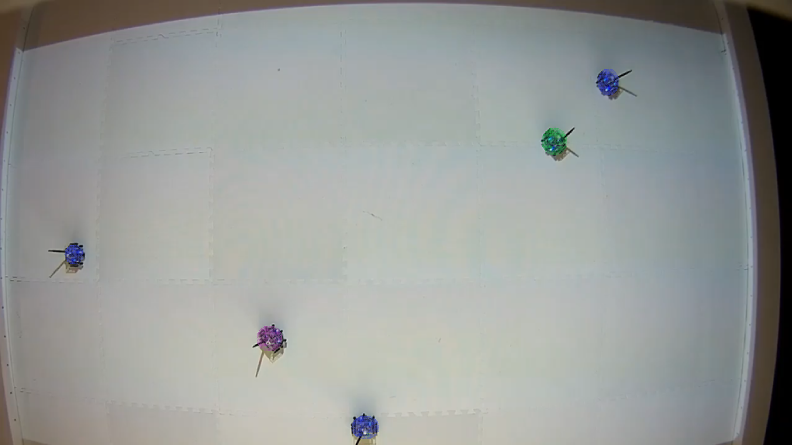}\hspace{-0.2cm}
         &\hspace{0.2cm}  
         \includegraphics[width=0.43\columnwidth]{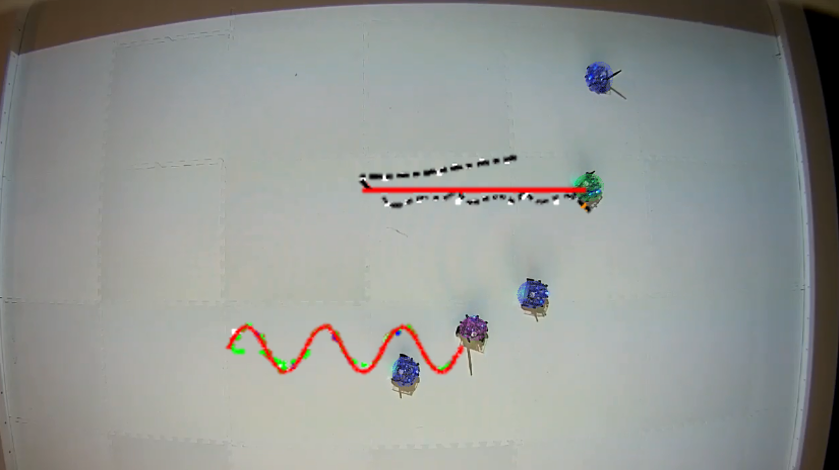}
    \end{tabular}
	\caption{Three herders herding two heterogeneous evaders in a time-varying path. The green evader is Inverse, and the purple evader is Exponential: (left) the initial configuration, (right) the final configuration with the desired and performed trajectories. More details in the supplementary video.}
	\label{fig:TV_robotarium}
\end{figure}

The experiment in Fig.~\ref{fig:first_impression} shows a similar behaviour to the simulations in Section~\ref{sec:simulations}.
The evaders try to evade the herders, going in the direction of lower density of herders. To tackle this, the closest herders surround the evaders to align with the other herders, which move away to modulate the interaction forces. These findings are reinforced with the experiments in Fig.~\ref{fig:second_impression}, where different combinations of number and evaders' dynamics are tested. 
Additionally, the herders successfully herd heterogeneous groups of evaders, Fig.~\ref{fig:second_impression}c, Fig.~\ref{fig:second_impression}d, and in the time-varying experiment in Fig.~\ref{fig:TV_robotarium}. 
Initially, the herders drive the evaders to initial positions and after that herd them in desired trajectories (in red). Despite the space limitations and the complex nonlinear repulsive dynamics, the evaders successfully follow the references.  


\section{Conclusions}\label{sec:conclusion}
This paper has addressed a novel control strategy to solve the herding problem in MRS. This strategy, based on numerical analysis theory, finds suitable herding actions even when, due to the complex nonlinearities of the herd, the control law is given by a set of implicit equations. 
To solve these equations the paper derives two design procedures. The Explicit Design develops an explicit continuous-time expansion of the system, and comes with formal proofs of convergence and fast execution time when the Jacobians are derived analytically; the Implicit Design leverages numerical methods to compute the action in discrete time, being an easy-to-implement approach which can leverage any standard numerical method. Both methods are flexible to the number of evaders and general with respect to their motion model.


\bibliographystyle{IEEEtran}
\bibliography{IEEEabrv,IEEEexample.bib}

\end{document}